\documentclass[a4paper]{scrartcl}
\usepackage[utf8]{inputenc}
\usepackage{amsmath,amsthm,amsfonts,amssymb}
\usepackage[bookmarksnumbered=true]{hyperref}
\usepackage{dsfont}
\usepackage{graphicx}
\usepackage{tikz-cd}
\usepackage{mathrsfs}
\usepackage[nottoc,numbib]{tocbibind}

\let\oldmarginpar\marginpar
\renewcommand\marginpar[1]{\-\oldmarginpar[\raggedleft\footnotesize #1]%
  {\raggedright\footnotesize #1}}

\newtheorem{thm}{Theorem}[section]  
\newtheorem*{thm*}{Theorem}

\newtheorem{lem}[thm]{Lemma}

\newtheorem*{rem}{Remark}

\newtheorem*{df}{Dirac-Frenkel Principle for Reduced Density Matrices}

\newcommand{\sone}{{\Sfrak_1}}
\newcommand{\stwo}{{\Sfrak_2}}
\newcommand{\sinf}{{\Bcal}}

\newcommand{\Sfrak}{\mathfrak{S}}
\newcommand{\hds}{\mathds{h}}
\newcommand{\ull}[1]{\underline{\underline{#1}}}

\newcommand{\Jcal}{\mathcal{J}}
\newcommand{\Wbb}{\mathbb{W}}
\newcommand{\gvac}{\Gamma_\text{vac}}
\newcommand{\Mcalqf}{\mathcal{M}}

\newcommand{\Ubb}{\mathbb{U}}
\newcommand{\Acal}{\mathcal{A}}
\newcommand{\Mcal}{\mathcal{M}}

\newcommand{\Scal}{\mathcal{S}}

\newcommand{\Ecal}{\mathcal{E}}
\newcommand{\projqf}{\operatorname{proj}}
\newcommand{\Gcal}{\mathcal{G}}
\newcommand{\gtqf}{\tilde\Gamma^\text{qf}_t}

\newcommand{\proj}{\operatorname{proj}}

\newcommand{\fock}{\mathcal{F}}		
\newcommand{\di}{{\textrm{d}}}		
\newcommand{\Ncal}{\mathcal{N}}		
\newcommand{\Vcal}{\mathcal{V}}		
\newcommand{\Hcal}{\mathcal{H}}		

\newcommand{\Zcal}{\mathcal{Z}}
\newcommand{\cc}[1]{\overline{#1}}	
\newcommand{\Rbb}{\mathbb{R}}		
\newcommand{\Cbb}{\mathbb{C}}		
\newcommand{\Nbb}{\mathbb{N}}		

\renewcommand{\Im}{\operatorname{Im}\,} 	
\newcommand{\id}{\mathds{1}}
\newcommand{\norm}[1]{\lVert#1\rVert}	
\newcommand{\tr}{\operatorname{tr}}

 \newcounter{notodos}
 
\newcommand{\wtK}{K}

\newcommand{\Bcal}{\mathcal{B}}

\newcommand{\Ycal}{\mathcal{Y}}

\newcommand{\oml}{\omega^{(\Lambda)}}
\newcommand{\gaml}{\gamma^{(\Lambda)}}
\newcommand{\Gaml}{\Gamma^{(\Lambda)}}
\newcommand{\alphl}{\alpha^{(\Lambda)}}

\newcommand{\Pl}{P_{\Lambda}}

\newcommand{\eps}{\varepsilon}

\newcommand{\dom}{\operatorname{dom}}

\begin{document}
\title{The Dirac--Frenkel Principle for Reduced Density Matrices, and the Bogoliubov--\\de\,Gennes Equations}

\author{Niels Benedikter, J\'er\'emy Sok, and Jan Philip Solovej}
\maketitle

\begin{abstract}
The derivation of effective evolution equations is central to the study of non-stationary quantum many-body systems, and widely used in contexts such as superconductivity, nuclear physics, Bose--Einstein condensation and quantum chemistry.

We reformulate the Dirac--Frenkel approximation principle in terms of reduced density matrices and apply it to fermionic and bosonic many-body systems. We obtain the Bogoliubov--de\,Gennes and Hartree--Fock--Bogoliubov equations, respectively. While we do not prove quantitative error estimates, our formulation does show that the approximation is optimal within the class of quasifree states. Furthermore, we prove well-posedness of the Bogoliubov--de\,Gennes equations in energy space and discuss conserved quantities.
\end{abstract}

\tableofcontents

\section{Introduction: Effective Evolution Equations}\label{sec:introduction}
The time evolution of the state $\psi_t$ of a system of $N$ quantum particles is described by the time-dependent Schr\"odinger equation
\begin{equation}\label{eq:SE}i \partial_t \psi_t = H_N \psi_t,\quad H_N = \sum_{i=1}^N h_i + \sum_{i<j} V(x_i-x_j),\end{equation}
where $\psi_t \in L^2_\text{a}(\Rbb^{3N})$ for spinless fermions (where the wave function is totally antisymmetric under any exchange of particles) and $\psi_t \in L^2_\text{s}(\Rbb^{3N})$ for bosons (where the wave function is totally symmetric under any exchange of particles).
In this generality, the Schr\"odinger equation models a vast range of physical systems, starting from nucleons in the atomic nucleus over electrons in semiconductors to stars for the fermionic theory, or Bose--Einstein condensates for the bosonic theory, depending on the choice of the one-particle Hamiltonian $h$ (we think of $h=-\Delta + V_\text{ext}(x)$ with some external potential $V_\text{ext}: \Rbb^3 \to \Rbb$; $h_i$ denotes this operator as acting in the variable $x_i \in \Rbb^3$) and the pair interaction $V: \Rbb^3 \to \Rbb$. Unfortunately, these systems also have an enormous number of degrees of freedom, making analytical and numerical solutions generally impossible. For this reason, there is a lot of interest in approximate theories (also  called \emph{effective evolution equations}), which contain fewer degrees of freedom and make analytical and numerical treatments possible. Of course, such theories do not achieve the broad validity of the Schr\"odinger equation and provide a good approximation only in specific physical regimes, which are mathematically modeled as scaling limits. In this paper, we discuss a geometric method for the derivation of effective evolution equations. This method, even though not the most convenient for proving quantitative error estimates for the obtained approximation, directly shows that the obtained equations are optimal as far as the available degrees of freedom permit. The method is also independent of any choice of scaling limit.

For this introduction we focus on (for simplicity of notation spinless) fermionic systems. The corresponding bosonic notions will be introduced in Sect.\ \ref{sec:bosons}. Here we deal only with \emph{pure} quasifree states. Only in Sect.\ \ref{sec:wellposedness} for the topic of well-posedness we also consider mixed states; we will always highlight explicitly when we talk about mixed states.

The most basic approximate theory of fermionic systems is obtained by restricting the Schr\"odinger equation to wave functions that are \emph{Slater determinants},
\[\psi(x_1,\ldots x_N) = (N!)^{-1/2} \det \big( f_{i}(x_j) \big)_{i,j=1}^N,\]
also denoted as the antisymmetrized tensor product $\psi = (N!)^{-1/2} f_{1}\wedge\cdots \wedge f_{N}$, where the one-particle wave functions $f_j$ constitute an orthonormal set in $L^2(\Rbb^3)$. The corresponding effective evolution equation for the Slater determinant  is given by the \emph{Hartree-Fock system} of $N$ nonlinear coupled PDEs for the one-particle wave functions:
\begin{equation}
 \label{eq:TDHFwave function}
 i \partial_t f_{i,t} = h f_{i,t} + \sum_{j=1}^N \left (V \ast \lvert f_{j,t}\rvert^2 \right) f_{i,t} - \sum_{j=1}^N \left (V \ast f_{i,t}\cc{f_{j,t}} \right) f_{j,t}.
\end{equation}
More conveniently, the Hartree-Fock equations can be formulated in terms of the one-particle reduced density matrix.

The \emph{one-particle reduced density matrix} of a state $\psi \in L^2(\Rbb^{3N})$ is defined as the non-negative trace-class operator $\gamma$ on $L^2(\Rbb^3)$ obtained by taking the partial trace over $N-1$ particles of the many-body density matrix $\lvert \psi \rangle \langle \psi \rvert$,
\[\gamma = N \tr_{2,\ldots N} \lvert \psi \rangle \langle \psi \rvert,\]
where we have chosen to normalize the one-particle reduced density matrix such that $\tr \gamma = N$.
If $\psi = (N!)^{-1/2} f_1 \wedge \ldots \wedge f_N$, we find the rank-$N$ projection $\gamma = \sum_{j=1}^N \lvert f_j\rangle\langle f_j\rvert$. Conversely, every rank-$N$ orthogonal projection specifies (uniquely up to a phase) a Slater determinant (just take the spectral decomposition of the projection to find the one-particle wave functions $f_j$).

If the one-particle wave functions have time dependence given by the Hartree-Fock equations \eqref{eq:TDHFwave function}, then $\gamma_t = \sum_{j=1}^N \lvert f_{j,t}\rangle\langle f_{j,t}\rvert$ satisfies the equivalent equation
\begin{equation}
 \label{eq:TDHFonepdm}
 i\partial_t \gamma_t = [h_\text{HF}(\gamma_t),\gamma_t], \quad h_\text{HF}(\gamma_t) = h + V\ast \rho_{\gamma_t} - X_V(\gamma_t),
\end{equation}
where $V\ast \rho_{\gamma_t} (x) = \int \di y\, V(x-y) \gamma_t(y,y)$ is a multiplication operator called the direct term, and $X_V(\gamma_t)(x,y) = V(x-y)\gamma_t(x,y)$ is the integral kernel of an operator called the exchange term.

To generalize Hartree-Fock theory to fermionic systems with pairing (as in superconductivity), we need to introduce Fock space. \emph{Fermionic Fock space} is defined as (the completion of) the direct sum
\[\fock_\text{a} := \Cbb \oplus \bigoplus_{n=1}^\infty L^2_\text{a}(\Rbb^{3n}),\]
i.\,e., elements of Fock space are sequences $\psi = \left(\psi^{(0)},\psi^{(1)}, \psi^{(2)},\ldots \right)$ with $\psi^{(0)}\in \Cbb$ and $\psi^{(n)} \in L^2_\text{a}(\Rbb^{3n})$, having finite norm $\norm{\psi}^2 = \sum_{j=0}^\infty \norm{\psi^{(n)}}^2$. Fock space is a Hilbert space with the scalar product $\langle \psi,\varphi\rangle = \sum_{n=0}^\infty \langle \psi^{(n)},\varphi^{(n)}\rangle$. Obviously, the $N$-particle space $L^2_\text{a}(\Rbb^{3N})$ can be considered as a subspace of Fock space $\fock_\text{a}$, and we frequently use this identification without distinguishing the vectors by notation. On Fock space, we introduce \emph{creation operators} $a^*(f)$ and \emph{annihilation operators} $a(f)$ (where $f \in L^2(\Rbb^3)$, a one-particle wave function) by (the hat indicates omission of the variable)
\[\begin{split}
  \left( a^*(f) \psi \right)^{(n)}(x_1,\ldots ,x_n) & = \frac{1}{\sqrt{n}}\sum_{j=1}^n (-1)^j f(x_j) \psi^{(n-1)}(x_1,\ldots ,\widehat{x_{j}},\ldots,x_n),\\
  \left( a(f) \psi \right)^{(n)}(x_1,\ldots ,x_n) & = \sqrt{n+1}\int \di x\, \cc{f(x)} \psi^{(n+1)}(x,x_1,\ldots,x_n).
  \end{split}
\]
They satisfy the \emph{canonical anti-commutation relations} (CAR), i.\,e.\ \[\{a(f),a(g)\} = 0,\ \{a^*(f),a^*(g)\}=0, \text{ and }\{a(f),a^*(g)\} = \langle f,g\rangle\] for all $f,g \in L^2(\Rbb^3)$. (The definition of the anti-commutator is $\{A,B\} = AB+BA$.) The vector $\Omega = \left(1,0,0, \ldots \right) \in \fock_\text{a}$ is called the \emph{vacuum state} and is in the kernel of all annihilation operators, $a(f)\Omega = 0$ for all $f \in L^2(\Rbb^3)$; it describes a system not containing any particles. It is convenient to  introduce the \emph{operator-valued distributions} $a^*_x$ and $a_x$ with the defining property that (in a weak sense, within expectation values)
\[a(f) = \int \di x\, a_x \cc{f(x)}, \quad a^*(f) = \int \di x\, a^*_x f(x).\]
They satisfy the formal canonical anti-commutation relations $\{a_x,a_y\}=0$, $\{a^*_x,a^*_y\}=0$, and $\{a_x,a^*_y\} = \delta(x-y)$.

The Hamiltonian is generalized to Fock space as
\begin{equation}\label{eq:numberconservinghamiltonian}H \psi = \left( H_n \psi^{(n)}\right)_{n=0}^\infty,\end{equation}
where $H_n$ denotes the first quantized Hamiltonian as given in \eqref{eq:SE}. The Hamiltonian $H$ can also be represented in terms of creation and annihilation operators by\footnote{The second quantization of $h$ is the operator $\di\Gamma(h)$, acting on the $n$-particle component $\psi^{(n)}$ of the Fock space vector $\psi$ as $\di\Gamma(h)\psi^{(n)} = \sum_{i=1}^n h_i \psi^{(n)}$. If $h$ has an integral kernel $h(x,y)$, then it can be written as $\di\Gamma(h) = \int \di x\di y\, h(x,y) a^*_x a_y$.}
\begin{equation}\label{eq:hamiltonian}
 H = \di\Gamma(h) + \frac{1}{2} \int \di x\di y\, V(x-y) a^*_x a^*_y a_y a_x.
\end{equation}
Restricted to $N$-particle states, this Hamiltonian agrees with $H_N$. For the geometric considerations in this paper, the explicit form of $H$ does not need to be specified, as long as it is a self-adjoint operator and conserves the number of particles. (A Hamiltonian conserves the number of particles if it commutes with the particle number operator $\Ncal$ defined by $\Ncal \psi = \left( n \psi^{(n)} \right)_{n=0}^\infty$. In particular all operators of the form \eqref{eq:numberconservinghamiltonian} conserve the particle number.) Only Sect.\ \ref{sec:wellposedness} refers to the particular Hamiltonian \eqref{eq:hamiltonian}.

Using the creation and annihilation operators, the definition of the \emph{one-particle reduced density matrix $\gamma$ is extended} to states $\psi \in \fock_\text{a}$ by defining it to have the integral kernel
\begin{equation}\label{eq:gamma}\gamma(x,y) := \langle \psi, a^*_y a_x \psi\rangle.\end{equation}
We can now give a simple proof that $\gamma \leq \id$: For all $f \in L^2(\Rbb^3)$, using the CAR,
\begin{align*}\langle f,\gamma f\rangle = \langle \psi, a^*(f) a(f) \psi\rangle & \leq \langle \psi, a^*(f)a(f) + a(f)a^*(f)\psi\rangle \\ & = \langle \psi, \{ a^*(f),a(f)\} \psi\rangle = \langle f,f\rangle.\end{align*}

Slater determinants are a special case of a class of more general states in Fock space called quasifree states (see e.\,g.\ \cite{Solovej} for a very readable introduction). The defining property of \emph{quasifree states} is that they are exactly those states $\psi \in \fock_\text{a}$ for which the \emph{Wick theorem} holds, i.\,e.\ expectation values of creation and annihilation operators can be reduced to the sum of the expectation values of all possible pairings of just two operators (with the sign being the sign of the corresponding pairing); for example
\begin{align*}\langle \psi, a^\natural_1 a^\natural_2 a^\natural_3 a^\natural_4 \psi\rangle & = \langle \psi,a^\natural_1 a^\natural_2 \psi \rangle \langle \psi, a^\natural_3 a^\natural_4 \psi \rangle - \langle \psi, a^\natural_1 a^\natural_3 \psi \rangle \langle \psi, a^\natural_2 a^\natural_4 \psi \rangle \\& \quad + \langle \psi,a^\natural_1 a^\natural_4 \psi \rangle \langle \psi, a^\natural_2 a^\natural_3\psi \rangle.\end{align*}
(Here we used the notation $a^\natural_j$ to denote an operator without specifying whether it is a creation or annihilation operator.)
Notice that the Wick theorem allows us to express any expectation value of creation and annihilation operators in a quasifree state purely in terms of the one-particle reduced density matrix $\gamma$ and the \emph{pairing density}
\begin{equation}\label{eq:alpha}\alpha(x,y) := \langle \psi, a_y a_x \psi\rangle.\end{equation}
Slater determinants are exactly those quasifree states for which the pairing density identically vanishes, $\alpha = 0$. Notice that for a general quasifree state $\gamma$ is not a rank-$N$ projection; instead a quasifree state always satisfies
\begin{equation}\label{eq:quasifree}
\gamma^2 - \gamma = \alpha \cc{\alpha} \quad \text{and}\quad \cc{\alpha}\gamma = \cc{\gamma} \cc{\alpha}.
\end{equation}
Given $\gamma$ and $\alpha$ satisfying \eqref{eq:quasifree}, there is a (up to a phase unique) quasifree state $\psi \in \fock_\text{a}$ such that \eqref{eq:gamma} and \eqref{eq:alpha} hold. 

There is a more compact way of writing the equations \eqref{eq:quasifree} by introducing the generalized one-particle reduced density matrix $\Gamma$. The \emph{generalized one-particle reduced density matrix} $\Gamma$ is an operator on $L^2(\Rbb^3) \oplus L^2(\Rbb^3)$ given by
\begin{equation}\label{eq:generalized1pdm} \Gamma = \left(\begin{array}{cc}
             \gamma & \alpha\\ -\cc{\alpha} & 1-\cc{\gamma}
            \end{array}\right).
\end{equation}
The characterization \eqref{eq:quasifree} of quasifree states is equivalent to $\Gamma$ being an orthogonal projection, $\Gamma^2 = \Gamma = \Gamma^*$. Any (not necessarily quasifree) generalized one-particle reduced density matrix has the property
\begin{equation}\label{eq:gen1pdm}
 0 \leq \Gamma \leq \id, \quad \text{and thus } \Gamma^2 \leq \Gamma.
\end{equation}

In the theory of superconductivity, the pairing density is interpreted as the wave function of electrons that have formed Cooper pairs, which in many ways behave like bosons. These Cooper pairs are seen as the carriers of the superconducting current that has attracted so much attention for its technological applicability in the dissipationless transport of electricity.

In this paper, our focus lies on the effective evolution equation obtained by restriction of the many-body evolution to quasifree states with pairing. This system of effective evolution equations is known as
the Bogoliubov--de\,Gennes equations
\begin{equation}
 \label{eq:TDBCS}
 \begin{split}
 i \partial_t \gamma_t & = [h_\text{HF}(\gamma_t),\gamma_t] - \Pi_V(\alpha_t) \cc{\alpha_t} - \alpha_t \Pi_V(\alpha_t)^*, \\
 i\partial_t \alpha_t & = h_\text{HF}(\gamma_t) \alpha + \alpha \cc{h_\text{HF}(\gamma_t)} + \Pi_V(\alpha_t)(1-\cc{\gamma_t}) - \gamma_t \Pi_V(\alpha_t),
 \end{split}
\end{equation}
with $h_\text{HF}(\gamma_t)$ as defined in \eqref{eq:TDHFonepdm} and the operator $\Pi_V(\alpha_t)$ defined through its integral kernel $\Pi_V(\alpha_t)(x,y) := V(x-y) \alpha_t(x,y)$ (notice that $\Pi_V(\alpha_t)^* = - \Pi_V(\cc{\alpha_t})$).
More compactly, $\gamma_t$ and $\alpha_t$ satisfy \eqref{eq:TDBCS} if and only if the generalized one-particle density matrix $\Gamma_t$ satisfies
\begin{equation}\label{eq:TDBCSgeneralized}i \partial_t \Gamma_t = [F_{\Gamma_t},\Gamma_t];\end{equation}
as in \cite{HS-BCS,Hainzletal} we use the generalized Hartree-Fock operator
\begin{equation}\label{eq:genhfop}F_{\Gamma_t} = \left( \begin{array}{cc}
                          h_\text{HF}(\gamma_t) & \Pi_V(\alpha_t)\\ \Pi_V(\alpha_t)^* & - \cc{h_\text{HF}(\gamma_t)}
                         \end{array}
 \right)\end{equation}
 on $L^2(\Rbb^3)\oplus L^2(\Rbb^3)$.
 The Bogoliubov--de\,Gennes equations for fermionic systems are sometimes also called the generalized Hartree-Fock equations or fermionic Hartree--Fock--Bogoliubov equations; the usual Hartree-Fock equations correspond to the Bogoliubov--de\,Gennes equations with $\alpha =0$. By restricting to physical regimes where direct and exchange term are negligible, and by including electron spin, one obtains the time-dependent BCS equations (named after Bardeen, Cooper, and Schrieffer), which describe the dynamics of electrons and Cooper pairs in superconductors.
 
 In the present paper, our goal is to
formulate a systematic approximation principle by which we can obtain the Bogoliubov--de\,Gennes equations from many-body quantum theory. The approximation principle we establish is a reformulation of the Dirac--Frenkel principle in the space of reduced density matrices, and it yields the equations sometimes called the quasifree reduction principle. Applying the quasifree reduction principle to the Hamiltonian \eqref{eq:hamiltonian}, one obtains the Bogoliubov--de\,Gennes equations. Afterward we study well-posedness and conserved quantities for the Bogoliubov--de\,Gennes equations.

While it is general knowledge that the quasifree reduction principle should be a consequence of the Dirac--Frenkel principle, we are not aware of a direct proof having appeared before; in particular the formulation of the Dirac--Frenkel principle in terms of reduced density matrices has not been given before. Among the advantages of our approach is that it shows that the obtained effective equations describe the optimal evolution possible within the approximation manifold.

\medskip

\noindent \emph{Earlier results.} The derivation of effective evolution equations for many-body systems has attracted a lot of attention in the community of mathematical physics and can be seen as a cornerstone of non-equilibrium statistical mechanics. Consequently, the literature is vast and we cannot claim to provide a complete overview. Let us say so much, that the geometric approximation principle on which we build in this paper goes back to the founding fathers of quantum mechanics \cite{Dirac, Frenkel}. A rigorous mathematical discussion and highly valuable presentation has been given in \cite{Lubich}.

The next step after the geometric derivation of the correct effective equation lies in the proof of convergence toward the effective equation and then the derivation of quantitative error bounds in given physical regimes modeled as scaling limits. This topic has attracted a lot of attention in recent years. For bosonic systems, many such results on the approximation of reduced density matrices have been proven for the mean-field model \cite{Mauser,Yau,RodSchlein,CLS,Pickl,Pickl2,Sohinger,Adami,Kirk,FKS,Anapolitanos} and the Gross--Pitaevskii model \cite{Bdos,EESY,ESY1,ESY2,ESY3,ESY4,PicklGP1,PicklGP2}. For fermionic systems, most results have only appeared in the last few years. The main regimes treated here are the mean-field regime on short time scales \cite{Golse,Frohlich2011}, the mean-field regime with slow variation of the effective interaction potential \cite{BBPPT,PP,Petrat}, and the combined mean-field/semi-classical limit for high-density systems \cite{NarnhoferSewell,SpohnVlasov,Elgart,BPS1,BPS2,BPSSJ3,RSSP}. The papers cited in this paragraph do not use the Dirac--Frenkel principle but other methods that have been specifically developed for many-body systems, like the BBGKY hierarchy, coherent states, a Schwinger-Dyson expansion or counting the number of particles well-described by the effective evolution equation. Some applications of the Dirac--Frenkel principle with explicit error estimates can be found in \cite{Lubich}; another example is \cite{Griesemer}.

In the context of proving convergence toward effective equations for bosonic systems in mean-field and Gross--Pitaevskii scaling limits, equations of the Hartree--Fock--Bogoliubov-type appear when considering second-order corrections beyond the above results on approximation of reduced density matrices, that is approximations in the norm of the many-body Hilbert space \cite{marcin1,marcin2,marcin3,marcin4,marcin5,marcin6}. The difference to our discussion here is that we are just interested in the approximation of reduced density matrices, not in the norm of the many-body Hilbert space, but instead we focus on ensuring optimality of the derived effective equations. Moreover, while the scaling regimes are crucial for obtaining convergence and quantitative error estimates, in the qualitative geometric approach that we use, it is not necessary to specify a particular scaling limit.

For fermionic systems, the derivation of quantitative error bounds in appropriate scaling limits for the Bogoliubov--de\,Gennes equations or even the BCS equations remains an open problem.

\medskip

\noindent \emph{Organization of the paper.} In Sect.\ \ref{sec:diracfrenkel} we shall recall the variational principle of Dirac and Frenkel, which is at the base of our paper. In Sect.\ \ref{sec:quasifreereduction} we recall the principle of quasifree reduction, which is less fundamental and less general than the Dirac--Frenkel principle, but more convenient for explicit calculations. In Sect.\ \ref{sec:hartreefock} we give a re-derivation of the Hartree-Fock equation as the simplest example to introduce our formulation of the Dirac--Frenkel principle in the space of one-particle reduced density matrices. In Sect.\ \ref{sec:BCS} we use our formulation of the Dirac--Frenkel principle for systems with pairing, yielding the time-dependent Bogoliubov--de\,Gennes equations. In Sect.\ \ref{sec:bosons}, we present the analogous formulation of the Dirac--Frenkel principle for bosonic systems. Finally, in Sect.\ \ref{sec:wellposedness} we discuss well-posedness of the time-dependent fermionic Bogoliubov--de\,Gennes equations.

\section{Approximation Principles}
In this section we first introduce the Dirac--Frenkel principle, which can be seen as the fundamental principle for deriving optimal effective evolution equations. Afterwards we introduce the principle of quasifree reduction which does not exhibit the optimality but leads to the same results as the Dirac--Frenkel principle and is calculationally simpler.
\subsection{The Dirac--Frenkel Variational Principle} \label{sec:diracfrenkel}
In this section we discuss the abstract formulation of the Dirac--Frenkel variational principle for the approximation of the time-dependent Schr\"odinger equation by projection onto a submanifold. The Dirac--Frenkel variational principle is particularly interesting for its clear geometrical content, which shows that the obtained equation on the submanifold is the optimal choice. We follow \cite[Chapter II]{Lubich}.

We consider the Schr\"odinger equation as an evolution equation in a complex Hilbert space $\mathcal{H}$. We use the convention that the scalar product is anti-linear in the first and linear in the second argument. The Hamiltonian $H$ is a self-adjoint operator on $\mathcal{H}$. The Schr\"odinger equation reads
\begin{equation}\label{eq:schrodinger}
 \partial_t \psi_t = \frac{1}{i}H \psi_t.
\end{equation}

Now consider a smooth (typically infinite dimensional) submanifold $\Mcal$ of $\mathcal{H}$. The tangent space of $\Mcal$ in the point $u \in \Mcal$ is denoted by $T_u\Mcal$; it consists of the derivatives $u'_0$ of all differentiable paths $t \mapsto u_t$ passing through $u_0=u$.

We are interested in approximating the solution $\psi_t$ of the Schr\"odinger equation with a path $u_t$ on the manifold $\Mcal$, assuming that initially $\psi_0 = u_0 \in \Mcal$. As pointed out already by Dirac, the path $t \mapsto u_t$ is to be chosen such that at every time $t$ the derivative $u'_t \in T_{u_t}\Mcal$ is as close as possible to $\frac{1}{i}H u_t$; in other words, the path is determined by choosing its derivative as the orthogonal projection of $\frac{1}{i}H u_t$ onto the tangent space:
\begin{equation}\label{eq:DFprinciple}
 \partial_t u_t = P(u_t) \frac{1}{i} Hu_t,
\end{equation}
$P(u_t)$ being the orthogonal projection from $\mathcal{H}$ onto $T_{u_t}\Mcal$. From this formulation it is clear that the Dirac--Frenkel principle yields the effective evolution equation which at every infinitesimal time step is optimal.

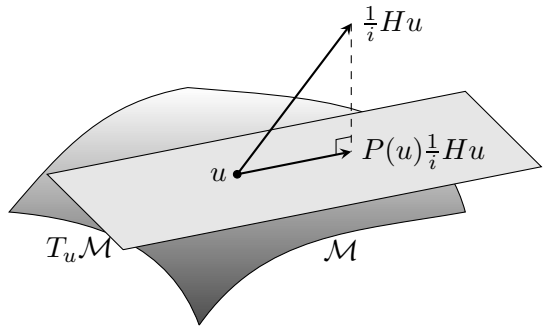
\begin{figure}\centering
\begin{tikzpicture}
 \draw [bottom color=white!30!black, top color=white] (3,1) to [out=55,in=190] (6.5,2.5) to [out=105,in=355] (2.85,4.15) to [out=190,in=50] (0.5,2.5) to [out=355,in=110] (3,1);
 \draw [fill=black!10!white](2,2) -- (1,3)--(6.5,4.1) --( 7.5,3.1) -- (2,2);
 \draw [thick,->,>=stealth] (3.5,3)--(5,5);
 \draw [thick,->,>=stealth] (3.5,3)--(5,3.3);
 \draw [fill]	(3.5,3) circle [radius=0.05];
 \draw [dashed] (5,5)--(5,3.3);
 \draw (5,3.5) -- (4.8,3.46) -- (4.8,3.26);
 \node [right] at (5,5) {$\frac{1}{i}Hu$};
 \node [left] at (2,2) {$T_u \Mcal$};
  \node [right] at (5,3.3) {$P(u)\frac{1}{i}Hu$};
  \node [left] at (3.5,3) {$u$};
  \node [right] at (4.5,2) {$\Mcal$};
\end{tikzpicture}
 \caption{The Dirac--Frenkel principle: Consider $u \in \Mcal$. At every ``time step'', the tangent $\frac{1}{i}Hu$ of the exact evolution  is orthogonally projected into the tangent space $T_u\Mcal$, yielding the optimal effective evolution in the approximation manifold $\Mcal$. {\footnotesize Figure following \cite{Lubich}.}}
\end{figure}

In the case of fermionic many-body systems, one typically chooses $\mathcal{H}= L^2_\text{a}(\Rbb^{3N})$ (for the case of no pairing) and $\Mcal$ as the set of $N$-particle Slater determinants.  While this approach does yield the time-dependent Hartree-Fock equations \eqref{eq:TDHFwave function}, it is not expected to ever do so with controllable errors: it is a general fact that in many-body systems, the norm of many-body wave functions as a measure of distance has unfortunate behavior as the number of particles grows. In fact, the quantitative derivation of effective equations for many-body systems in appropriate scaling limits is typically proven in terms of the trace norm or Hilbert--Schmidt norm of reduced density matrices (see the overview of results in Sect.\ \ref{sec:introduction}). For this reason, and to make the connection to the principle of quasifree reduction, in this paper we formulate the Dirac--Frenkel principle in the space of one-particle density matrices. This formulation will be calculationally very convenient when we consider quasifree states with pairing.

\subsection{The Principle of Quasifree Reduction}\label{sec:quasifreereduction}
The principle of quasifree reduction appears to be the computationally most accessible principle for deriving effective equations for many-body quantum systems. It applies to the particular case were the approximation manifold is given by a class of quasifree states. Typically it is formulated directly in the language of reduced density matrices.

The principle of quasifree reduction asserts that for fermionic systems the effective evolution equations are 
\begin{equation}\label{eq:quasifreereduction}\begin{split}
   \partial_t \gamma_t(x,y) & = \langle \psi_t^{\text{qf}}, [a^*_y a_x, \frac{1}{i}H] \psi_t^{\text{qf}} \rangle,\\
   \partial_t \alpha_t(x,y) & = \langle \psi_t^{\text{qf}}, [a_y a_x,\frac{1}{i}H]\psi_t^{\text{qf}} \rangle,
  \end{split}
\end{equation}
where $\psi_t^{\text{qf}} \in \fock_\text{a}$ is the quasifree state uniquely (up to a phase) assigned to $\gamma_t$ and $\alpha_t$.

For bosonic systems the equations proposed as the quasifree reduction principle, including a condensate $\varphi_t \in L^2(\Rbb^3)$, are the following \cite{Bachetal}
\begin{equation}\label{eq:bosonicqfreduction}
 \begin{split}
  \partial_t \varphi_t(x) & = \langle \psi^\text{bog}_t, [a_x,\frac{1}{i}H] \psi^\text{bog}_t \rangle,\\
  \partial_t \gamma_t(x,y) & = \langle \psi^\text{bog}_t, [a^*_y a_x,\frac{1}{i}H] \psi^\text{bog}_t \rangle, \\
  \partial_t \alpha_t(x,y) & = \langle \psi^\text{bog}_t, [a_y a_x,\frac{1}{i}H] \psi^\text{bog}_t \rangle, \\
 \end{split}
\end{equation}
where $\psi^\text{bog}_t$ in the bosonic Fock space is the Bogoliubov state associated with $(\varphi_t,\gamma_t,\alpha_t)$ (see \eqref{eq:bogoliubovstate} for the definition).

While this principle is easy to formulate, calculationally accessible and has been frequently used in many contexts (e.\,g.\ very recently to derive the Hartree--Fock--Bogoliubov equations for bosonic systems \cite{Bachetal} and the Bogoliubov--de\,Gennes equations for fermionic systems \cite{ChenSigal}), it is not completely obvious that it is a consequence of the more fundamental Dirac--Frenkel principle. Maybe more severely, it is not at all clear whether the quasifree reduction principle yields the optimal approximation possible within the manifold of quasifree states. 

In the present paper, we prove that the principle of quasifree reduction does follow directly from the Dirac--Frenkel principle, in particular showing that it yields the optimal effective evolution equations.

\section{Derivation of the Quasifree Reduction Principle for Fermions}
In this section we derive the principle of quasifree reduction from a re-formulation of the Dirac--Frenkel principle. We first sketch the instructive case of systems without pairing before generalizing to fermionic systems that also exhibit pairing.
\subsection{Fermionic Systems without Pairing}\label{sec:hartreefock}
Here we shall warm up with the case of no pairing ($\alpha = 0$), i.\,e.\ giving a derivation of the standard Hartree-Fock equation. To this end we shall formulate the Dirac--Frenkel principle in terms of the one-particle reduced density matrix and from there, derive the principle of quasifree reduction. To our knowledge, this is the first formulation of the Dirac--Frenkel principle in terms of reduced density matrices. Notice that the derivation of the quasifree reduction principle does not make use of any particular form of the Hamiltonian; we assume only that it commutes with the particle number operator (i.\,e.\ the number of particles is conserved along the many-body evolution).

\medskip

\noindent We consider the Schr\"odinger equation in the space $L^2_\text{a}(\Rbb^{3N})$, i.\,e.\ describing a fermionic system of $N$-particles. The many-body evolution $t \mapsto \psi_t$ induces an evolution of the associated one-particle reduced density matrix $\gamma_{\psi_t}$, which satisfies
\begin{equation}
\label{eq:manybodyevolution}
\partial_t \gamma_{\psi_t} = N \tr_{2,\ldots N} \left[\frac{1}{i}H,\lvert\psi_t\rangle\langle \psi_t\rvert\right].\end{equation}
The one-particle reduced density matrix is a non-negative operator on $L^2(\Rbb^3)$. Since our system has a finite number of particles, $\tr \gamma_{\psi_t} = N$, the one-particle reduced density matrix is a Hilbert--Schmidt operator. We thus choose the ambient Hilbert space in which $\gamma_{\psi_t}$ lives to be
\[\Hcal := \{\gamma \in \mathfrak{S}_2(L^2(\Rbb^3)): \gamma = \gamma^*\}.\]
Due to the condition of self-adjointness this is only a \emph{real-linear} (instead of complex-linear) space; in the following all spaces are real-linear only.
Corresponding to Slater determinants in the wave function picture, we choose our approximation manifold to be given by orthogonal projections,
\[\Mcal := \{\gamma \in \Hcal: \gamma^2 = \gamma\}.\]
This is an infinite dimensional Hilbert submanifold of the self-adjoint Hilbert--Schmidt operators. The effective evolution equation is to be found in $\Mcal$. This is achieved in the optimal way by applying the Dirac--Frenkel principle reformulated in terms of the one-particle reduced density matrix.

\begin{df} The effective evolution equation within the submanifold $\Mcal$ is given by
\begin{equation}\label{eq:dfprinc}\partial_t \gamma_t  = \proj(\gamma_t) N \tr_{2,\ldots N} \big[\frac{1}{i}H, \lvert \psi^{\text{qf}}_t \rangle \langle \psi^{\text{qf}}_t\rvert\big],\end{equation}
where $\psi^\text{qf}_t$ is the state uniquely (up to the phase) associated with $\gamma_t$, and $\proj(\gamma_t): \Hcal \to T_{\gamma_t}\Mcal$ is the projection onto the tangent space in the point $\gamma_t$.
\end{df}

We start by determining the tangent space and the projection onto the tangent space.
\begin{lem}[Tangent Space, No Pairing]\label{lem:tangentspacenopairing}
The tangent space in a point $\gamma \in \Mcal$ is
\[T_\gamma \Mcal = \{A \in \Hcal: \gamma A\gamma =0 = (1-\gamma)A(1-\gamma)\}.\]
The orthogonal projection from $\Hcal$ onto $T_\gamma\Mcal$ is given by
\[\proj(\gamma): A \mapsto \gamma A(1-\gamma) + (1-\gamma)A\gamma = [[A,\gamma],\gamma].\] 
\end{lem}
\begin{proof}
 Let $A \in T_\gamma\Mcal$. By definition there exists a differentiable curve $t \mapsto \gamma_t$ in $\Mcal$ such that $\gamma_0 = \gamma$ and $\gamma'_0 = A$. (By definition of differentiability in the norm of the ambient Hilbert space $\Hcal$, $A$ is a Hilbert--Schmidt operator.) Taking the derivative of the projection condition $\gamma_t^2 = \gamma_t$, we find $A \gamma + \gamma A = A$. Multiplying from the left and right by $\gamma$, we get $2 \gamma A \gamma = \gamma A \gamma$, so $\gamma A\gamma = 0$. Furthermore, by multiplying it from the left and the right by $(1-\gamma)$, we get $0 = (1-\gamma)A(1-\gamma)$. Trivially $A^* = A$ since $\gamma_t^* = \gamma_t$.
 
 Conversely, let $(1-\gamma)A(1-\gamma) = 0$ and $\gamma A\gamma = 0$. Then $B := [A,\gamma]$ is anti-self-adjoint and Hilbert--Schmidt, so $e^{tB}$ is a unitary. Now let $\gamma_t := e^{tB} \gamma e^{-tB}$. This is a curve of orthogonal projections with $\gamma_0 = \gamma$. Its derivative at zero is
 \[\gamma'_0 = [B,\gamma] = A\gamma - 2 \gamma A \gamma + \gamma A = (1-\gamma)A\gamma + \gamma A(1-\gamma) = A,\]
 showing that $A \in T_\gamma\Mcal$.
 
 It is easy to see that $\proj(\gamma)$ is a projection and has the claimed image.
\end{proof}
The manifold of orthogonal projections $\Mcal$ has several connected components, corresponding to the value $\tr \gamma \in \Nbb$. Clearly, any differentiable curve always stays within the same connected component, so we do not have to worry about this.

\begin{thm}[Derivation of the Quasifree Reduction Principle, No Pairing]\label{thm:nopairing}
The Dirac--Frenkel principle \eqref{eq:dfprinc}, with $\Hcal$ and $\Mcal$ as chosen above, is equivalent to the quasifree reduction principle without pairing ($\alpha = 0$):
 \begin{equation}
 \label{eq:thm1}
 \partial_t \gamma_t  = N \tr_{2,\ldots N} \big[\frac{1}{i}H, \lvert \psi^{\text{qf}}_t \rangle \langle \psi^{\text{qf}}_t\rvert\big],\end{equation}
 where $\psi^\text{qf}_t$ is the Slater determinant uniquely (up to the phase) associated with $\gamma_t$. (Equation \eqref{eq:thm1} is the same as equation \eqref{eq:quasifreereduction} but written without the use of operator-valued distributions.)
\end{thm}
We refer the reader to the more general proof of Theorem \ref{thm:withpairing}.

\medskip

\noindent One may convince oneself that \eqref{eq:thm1} indeed yields the Hartree-Fock equation \eqref{eq:TDHFonepdm} when evaluating the expectation value on the r.\,h.\,s.\ with the many-body Hamiltonian from \eqref{eq:hamiltonian}, using the canonical anti-commutation relations and the Wick theorem.

\medskip

 \noindent Having \eqref{eq:TDHFonepdm} at hand, we can also obtain the quasifree reduction principle from the Dirac--Frenkel principle for reduced densities as follows. Clearly \eqref{eq:TDHFonepdm} implies that $\gamma_t^2$ also satisfies the Hartree-Fock equation, 
 \begin{equation}\label{eq:TDBCSgammasquare}i \partial_t\, \gamma_t^2 = [h_\text{HF}(\gamma_t), \gamma_t^2].\end{equation}
 So if we have a projection as initial data, $\gamma_0^2 = \gamma_0$, assuming uniqueness, we conclude that $\gamma_t^2 = \gamma_t$ for all times $t$.
 
 Alternatively, we could argue that the Hartree-Fock equation preserves the spectrum of $\gamma_t$, which also implies $\gamma_t^2 = \gamma_t$ for all times.
 
 Either way, we conclude that the derivative is in the tangent space of $\Mcal$, which  makes the projection in the Dirac--Frenkel principle trivial and yields the quasifree reduction principle. However:
 \begin{itemize}
  \item This argument uses \eqref{eq:TDHFonepdm} which is obtained by explicitly evaluating the quasifree reduction principle. Using only the equations of the quasifree reduction principle \eqref{eq:quasifreereduction}, there is no easy way to formulate \eqref{eq:TDBCSgammasquare}; in fact, a direct verification that \eqref{eq:quasifreereduction} stays within $\Mcal$ seems complicated to us.
  \item The argument depends on the choice of the one-particle Hamiltonian $h$ and regularity and decay of the interaction potential $V$ and the initial data. For the initial value problem with pairing ($\alpha_0 \neq 0$), uniqueness or conservation of the spectrum are by themselves non-trivial problems, see Sect.\ \ref{sec:wellposedness}. 
  \end{itemize} 
  Our derivation does not require any specification of the Hamiltonian beyond its existence as a self-adjoint, particle number conserving, operator. Furthermore, our geometric approach makes it clear that the quasifree reduction principle is the \emph{optimal} approximation within the set of quasifree states.
  (Also $\partial_t \gamma_t(x,y) = 2\,\langle \psi_t^{\text{qf}}, [a^*_y a_x, \frac{1}{i}H] \psi_t^{\text{qf}} \rangle$  or  $\partial_t \gamma_t=0$ would be an evolution in the manifold of quasifree states---but far from being the optimal approximation to the many-body problem.)

\subsection{Fermionic Systems with Pairing}\label{sec:BCS}
We now extend our formulation of the Dirac--Frenkel principle to derive the approximation of fermionic many-body systems by quasifree states with pairing, namely the Bogoliubov--de\,Gennes equations. As before, the derivation of the quasifree reduction principle does not require any particular form of the Hamiltonian; our only assumption is that it conserves the number of particles.

\medskip

\noindent The geometry becomes very similar to the case of no pairing by using the generalized one-particle reduced density matrix $\Gamma$. Thus the manifold of quasifree states with pairing can be described by $\Gamma^2 = \Gamma$ and the block structure \eqref{eq:generalized1pdm}, which however comes `for free' since it is present in all generalized one-particle reduced density matrices (in particular also in the one derived from the many-body Schr\"odinger equation).

Let us be a bit more precise and define the involved spaces. First of all notice that, due to the form \eqref{eq:generalized1pdm}, we have
\[\tr \Gamma^*\Gamma = \tr \left(\begin{array}{cc}
             \gamma & \alpha\\ -\cc{\alpha} & \id-\cc{\gamma}
            \end{array}\right)^* \left(\begin{array}{cc}
             \gamma & \alpha\\ -\cc{\alpha} & \id-\cc{\gamma}
            \end{array}\right) = \tr \id = \infty;\]
the generalized one-particle density matrix is not a Hilbert--Schmidt operator. We remedy this problem by considering the generalized one-particle reduced density matrix as a point in an affine space, and the approximation manifold as a submanifold of this affine space. Let us denote by
\[\gvac := \left(\begin{array}{cc}
             0 & 0\\ 0 & \id
            \end{array}\right)\]
the generalized one-particle reduced density matrix of the vacuum $\Omega \in \fock$. Then any generalized one-particle reduced density matrix can be written as
\[\Gamma = \left(\begin{array}{cc}
             0 & 0\\ 0 & \id
            \end{array}\right) + \left(\begin{array}{cc}
             \gamma & \alpha\\ -\cc{\alpha} & -\cc{\gamma}
            \end{array}\right) =: \gvac + \vec{\Gamma}.\]
Every generalized one-particle reduced density matrix satisfies $\Gamma^2 \leq \Gamma$, which implies $\gamma^2 - \alpha \cc{\alpha} \leq \gamma$ (only for quasifree states we had equality here); thus
\[\tr \vec\Gamma^{*} \vec\Gamma = \tr (\gamma^2 - \alpha\cc{\alpha}) + \tr (\cc{\gamma}^2 - \cc{\alpha}\alpha) = 2 \tr (\gamma^2 - \alpha\cc{\alpha}) \leq 2 \tr \gamma,\]
which is twice the expected number of particles and as such assumed to be finite. The expected number of particles is trivially conserved along the many-body evolution since we assume the Hamiltonian to commute with the particle number operator; it is typically also conserved along the effective evolution, c.\,f.\ Lemma \ref{lem:conserved_spectrum_and_quantities}, so it is justified to take $\vec\Gamma$ as a Hilbert--Schmidt operator. Let us therefore introduce the affine space
\[\Acal := \gvac + \vec\Acal, \quad \vec\Acal := \{\vec\Gamma \in \mathfrak{S}_2(L^2(\Rbb^3)\oplus L^2(\Rbb^3)): \vec\Gamma = \vec\Gamma^{*} \}.\]
Similar to the no-pairing case, $\vec\Acal$ is a real-linear space.

Now notice that the requirement of having the block structure of $\Gamma$ in terms of $\gamma$ and $\alpha$ as in \eqref{eq:generalized1pdm} can be rewritten\footnote{There is a subtlety here: Not only $\Gamma = \left( \begin{array}{cc}\gamma&\alpha\\-\cc{\alpha}&1-\cc{\gamma} \end{array}\right)$ satisfies the equation $\Gamma +\Jcal\Gamma\Jcal = \id$, but so does also $\left( \begin{array}{cc}1-\gamma&\alpha\\-\cc{\alpha}&\cc{\gamma} \end{array}\right)$. The latter one however is not a Hilbert--Schmidt perturbation of $\gvac$ and thus not a solution within $\Acal$; in fact it corresponds to a state formally obtained from $\Gamma$ by a particle-hole transformation replacing the vacuum by an infinite number of fermions filling up all the Hilbert space.} as the condition
\begin{equation}\label{eq:blockcond}\Gamma + \Jcal \Gamma \Jcal = \mathds{1}, \text{ where } \Jcal = \left( \begin{array}{cc}0 & J\\J & 0\end{array}\right): L^2(\Rbb^3)\oplus L^2(\Rbb^3) \to L^2(\Rbb^3)\oplus L^2(\Rbb^3)\end{equation}
with $J: L^2(\Rbb^3) \to L^2(\Rbb^3)$ being the anti-linear operator of complex conjugation. So we can think of the evolution of the many-body generalized one-particle reduced density matrix as living in the affine subspace of $\Acal$ given by
\[\Acal_{-} := \left\{\Gamma \in \Acal: \Gamma + \Jcal \Gamma \Jcal = \mathds{1}\right\}.\]
(But not every $\Gamma \in \Acal_{-}$ is the generalized one-particle reduced density matrix of a Fock space vector.)

The approximation manifold is again given by the generalized one-particle density matrices corresponding to quasifree states:
\begin{equation}\label{eq:approximationmanifold}
\Mcalqf := \left\{\Gamma \in \Acal: \Gamma + \Jcal \Gamma \Jcal = \mathds{1},\ \Gamma^2 = \Gamma\right\}.\end{equation}
So compared to the no-pairing case not much has changed---the only additional complication is that we have to impose the block structure of $\Gamma$ in terms of $\gamma$ and $\alpha$. Luckily, this block structure is present in any generalized one-particle reduced density matrix including the one of the many-body evolution. So the many-body evolution describes a curve in the affine subspace $\Acal_{-}$, of which $\Mcal$ is a submanifold.

To provide a characterization of the tangent space, we also introduce as an auxiliary space the manifold of projections which do not necessarily have the block structure
 \begin{equation}
 \Mcal^\text{aux} := \left\{\Gamma \in\Acal: \Gamma^2 = \Gamma\right\}.\end{equation}
Notice that, since $\Acal$ is an affine space, $T_\Gamma\Acal = \vec\Acal$ for any $\Gamma \in \Acal$.
\begin{lem}[Tangent Space, With Pairing]\label{lem:pairingtangentspace}
For $\Gamma$ a point in the manifolds $\Mcal^\text{aux}$, $\Acal_{-}$ or $\Mcalqf$, respectively, let us introduce the following projections:
\begin{enumerate}
 \item onto the tangent space of projection operators \begin{equation}\label{eq:lem1}\proj^\text{aux}(\Gamma): \vec\Acal \to T_\Gamma\Mcal^\text{aux}, \quad \Xi \mapsto \Gamma \Xi (1-\Gamma) + (1-\Gamma)\Xi \Gamma = [[\Xi,\Gamma],\Gamma],\end{equation}
 \item onto the tangent space of the affine subspace with the block structure
 \begin{equation}\label{eq:lem2}\proj_{-}(\Gamma): \vec\Acal \to T_\Gamma\Acal_{-}, \quad \Xi \mapsto \frac{1}{2}\left( \Xi - \Jcal \Xi \Jcal \right),\end{equation}
 \item and onto the tangent space of quasifree states $\projqf(\Gamma): \vec\Acal \to T_\Gamma\Mcalqf$.
\end{enumerate}

\medskip

\noindent Then, for $\Gamma \in \Mcalqf$, we have
\[\projqf(\Gamma) = \proj_{-}(\Gamma) \proj^\text{aux}(\Gamma) = \proj^\text{aux}(\Gamma) \proj_{-}(\Gamma),\] and
\begin{equation}\label{eq:projectiononqf}
 \proj(\Gamma)\restriction_{T_\Gamma\Acal_{-}} = \proj^\text{aux}(\Gamma)\restriction_{T_\Gamma\Acal_{-}}.
\end{equation}
\end{lem}
\begin{proof}The projection onto the tangent space of projection operators \eqref{eq:lem1} is known from Lemma \ref{lem:tangentspacenopairing}.

Since $\Acal_{-}$ is an affine subspace, we can simply take the derivative of the defining equation to find
\[T_\Gamma \Acal_{-} = \left\{ \Xi \in \vec\Acal: \Xi + \Jcal \Xi \Jcal = 0 \right\}. \]
It is easy to check that the formula \eqref{eq:lem2} defines an orthogonal projection, maps into $T_\Gamma \Acal_{-}$ and is surjective onto $T_\Gamma \Acal_{-}$; therefore it is actually \emph{the} orthogonal projection onto $T_\Gamma \Acal_{-}$.

It is simple to check that $\proj_{-}(\Gamma) \proj^\text{aux}(\Gamma) = \proj^\text{aux}(\Gamma) \proj_{-}(\Gamma)$. Notice that $\Mcalqf \subset \Acal_{-}$, so $T_\Gamma \Mcalqf$ is a linear subspace of $T_\Gamma \Acal_{-}$ (for $\Gamma \in \Mcalqf$).  Thus, using Lemma \ref{lem:tangentspacenopairing},
 \[\begin{split}T_\Gamma \Mcalqf & = \left\{ \Xi \in T_\Gamma \Acal_{-}: \Gamma \Xi \Gamma =0 = (1-\Gamma) \Xi (1-\Gamma) \right\}\\
    &= \left\{ \Xi \in \vec\Acal: \Gamma \Xi \Gamma =0 = (1-\Gamma) \Xi (1-\Gamma) \text{ and } \Xi + \Jcal \Xi \Jcal =0 \right\}.
\end{split}\]
 Let $P := \proj_{-}(\Gamma) \proj^\text{aux}(\Gamma)$. Obviously $P$ is an orthogonal projection. It is easy to verify that it maps into $T_\Gamma\Mcal$ and is surjective onto $T_\Gamma\Mcal$; therefore $\proj(\Gamma) = P$.

Now let $A \in T_\Gamma\Acal_{-}$. Then $\proj^\text{aux}(\Gamma) A = \proj^\text{aux}(\Gamma) \proj_{-}(\Gamma) A = \proj(\Gamma) A$, so \eqref{eq:projectiononqf} holds. 
\end{proof}

We can now derive the quasifree reduction principle from the Dirac--Frenkel principle.

\begin{thm}[Derivation of the Quasifree Reduction Principle, With Pairing]\label{thm:withpairing}
 The effective equation for the generalized one-particle reduced density matrix $\Gamma_t$ obtained by applying the Dirac--Frenkel principle to the many-body evolution with $\Mcal$ and $\Acal$ as chosen above yields the principle of quasifree reduction
 \begin{equation}\label{eq:quasifreethm} \langle F_1, \left(\partial_t \Gamma_t\right) F_2\rangle = \langle\psi^{\text{qf}}_t, [A^*(F_2) A(F_1),\frac{1}{i}H] \psi^{\text{qf}}_t\rangle \quad \forall F_1, F_2 \in L^2(\Rbb^3) \oplus L^2(\Rbb^3),\end{equation}
 where $\psi^{\text{qf}}_t$ is the quasifree state uniquely (up to its phase) assigned to $\Gamma_t$. (Equation \eqref{eq:quasifreethm} is a compact way of writing \eqref{eq:quasifreereduction}, avoiding the use of operator-valued distributions by testing against $F_1$ and $F_2$.)
\end{thm}
\begin{proof}
The proof uses some theory of Bogoliubov transformations, for which we recommend \cite[Chapters 9 and 10]{Solovej} as a reference. (For the no-pairing case, the Bogoliubov transformation is a simple particle-hole transformations, see, e.\,g., \cite{BPS1,BPS2}.)

Recall that the generalized one-particle reduced density matrix $\Gamma_\psi$ of a Fock space vector $\psi$ is, avoiding the use of operator-valued distributions by testing against $F_1, F_2 \in L^2(\Rbb^3) \oplus L^2(\Rbb^3)$, given by
\begin{equation}\label{eq:defgen1pdm}\langle F_1, \Gamma_\psi F_2\rangle_{L^2(\Rbb^3) \oplus L^2(\Rbb^3)} = \langle \psi, A^*(F_2) A(F_1) \psi \rangle_{\fock},\end{equation}
where the generalized creation and annihilation operators are
\[A(\left(\!\! \begin{array}{c}f\\ g\end{array}\!\! \right)) := a(f) + a^*(\cc{g})\]
and
\[A^*(\left( \!\!\begin{array}{c}f\\ g\end{array}\!\! \right)) := a^*(f) + a(\cc{g}), \quad \text{for } \left(\!\! \begin{array}{c}f\\ g\end{array} \!\!\right) \in L^2(\Rbb^3)\oplus L^2(\Rbb^3).\]
So for $\psi_t$ being the solution of the many-body Schr\"odinger equation, the associated generalized one-particle reduced density matrix satisfies
 \[\langle F_1, \left(\partial_t \Gamma_t^{\text{MB}}\right) F_2\rangle_{L^2(\Rbb^3)\oplus L^2(\Rbb^3)} = \langle \psi_t, [A^*(F_2) A(F_1),\frac{1}{i}H] \psi_t \rangle_{\fock}.\]
 Notice that the r.\,h.\,s., like the derivative of any differentiable curve of generalized one-particle reduced density matrices, lies in $T_\Gamma\Acal_{-}$. According to the Dirac--Frenkel principle, we have to project it onto the tangent space of quasifree states.
 We apply the projection as given by \eqref{eq:projectiononqf} and \eqref{eq:lem1} to get
\[\begin{split}\langle F_1, \left( \partial_t \Gamma_t\right) F_2\rangle 
& = \langle \psi^{\text{qf}}_t , [A^*((1-\Gamma_t)F_2) A(\Gamma_t F_1),\frac{1}{i}H] \psi^{\text{qf}}_t \rangle_{\fock} \\
& \quad + \langle \psi^{\text{qf}}_t , [A^*(\Gamma_t F_2) A((1-\Gamma_t)F_1),\frac{1}{i}H] \psi^{\text{qf}}_t \rangle_{\fock},\end{split}\]
where $\psi^{\text{qf}}_t$ is the quasifree state uniquely assigned to $\Gamma_t$.
Comparing to the quasifree reduction principle \eqref{eq:quasifreethm}, we see that we simply have to show that $\langle \psi^{\text{qf}}_t , [A^*((1-\Gamma_t)F_2) A((1-\Gamma_t)F_1),\frac{1}{i}H] \psi^{\text{qf}}_t \rangle = 0$ and then also that $\langle \psi^{\text{qf}}_t , [A^*(\Gamma_t F_2) A(\Gamma_t F_1),\frac{1}{i}H] \psi^{\text{qf}}_t \rangle = 0$.

Since $\psi^{\text{qf}}_t$ is a quasifree state, it can be written in terms of an implementable Bogoliubov map ${\Vcal_t}: L^2(\Rbb^3) \oplus L^2(\Rbb^3) \to L^2(\Rbb^3) \oplus L^2(\Rbb^3)$ as $\psi^{\text{qf}}_t = \Ubb_{{\Vcal_t}} \Omega$ ($\Ubb_{\Vcal_t}$ being the unitary implementation in Fock space). Take any $F_1, F_2 \in L^2(\Rbb^3)\oplus L^2(\Rbb^3)$. Using the property $\Ubb_{\Vcal_t}^* A(F) \Ubb_{\Vcal_t} = A({\Vcal_t}^{-1}F)$ of the Bogoliubov map and recalling \eqref{eq:defgen1pdm}, we calculate
\[\begin{split}
   \langle F_1, \Gamma_t F_2\rangle_{L^2\oplus L^2} & = \langle \Ubb_{\Vcal_t} \Omega, A^*(F_2) A(F_1) \Ubb_{\Vcal_t} \Omega \rangle_\fock  \\ &= \langle \Omega, A^*({\Vcal_t}^{-1} F_2) A({\Vcal_t}^{-1} F_1) \Omega \rangle_\fock \\
   & = \langle {\Vcal_t}^{-1} F_1, \gvac {\Vcal_t}^{-1} F_2\rangle_{L^2\oplus L^2} = \langle F_1, {\Vcal_t} \gvac {\Vcal_t}^{-1} F_2\rangle_{L^2\oplus L^2},
  \end{split}
\]
so we obtain
\[{\Vcal_t}^* \Gamma_t {\Vcal_t} = \gvac = \left( \begin{array}{cc}0 & 0\\ 0& \id\end{array} \right).\]
Using this last identity we calculate that 
\[
\begin{split}
& \langle \psi^{\text{qf}}_t, \left[ A^*(\Gamma_t F_2) A(\Gamma_t F_1), \frac{1}{i}H \right] \psi^{\text{qf}}_t \rangle\\
& = \langle \Omega, \left[ A^*({\Vcal_t}^{-1}\Gamma_t F_2) A( {\Vcal_t}^{-1}\Gamma_t F_1) , \frac{1}{i}\Ubb_{\Vcal_t}^* H  \Ubb_{\Vcal_t}\right]  \Omega \rangle \\
& = \langle \Omega, \left[ A^*(\gvac{\Vcal_t}^{-1} F_2) A( \gvac {\Vcal_t}^{-1} F_1) , \frac{1}{i}\Ubb_{\Vcal_t}^* H  \Ubb_{\Vcal_t}\right]  \Omega \rangle \\
& = \langle \Omega, \left[ A^*(\left(\!\! \begin{array}{c}0 \\ \tilde g_2\end{array} \!\!\right)) A( \left(\!\! \begin{array}{c}0\\ \tilde g_1\end{array}\!\! \right)) , \frac{1}{i}\Ubb_{\Vcal_t}^* H  \Ubb_{\Vcal_t}\right]  \Omega \rangle
\end{split}
\]
where we have introduced the notation ${\Vcal_t}^{-1} F_i =: \tilde F_i =: \left(\!\! \begin{array}{c}\tilde f_i \\ \tilde g_i\end{array}\!\! \right)$ for $i \in \{1,2\}$.
 Now
\[
\begin{split}
& \langle \psi^{\text{qf}}_t, \left[ A^*(\Gamma_t F_2) A(\Gamma_t F_1), \frac{1}{i}H \right] \psi^{\text{qf}}_t \rangle\\
& = \langle \Omega, \left[ a( \cc{\tilde g_2}) a^*(\cc{\tilde g_1} ) , \frac{1}{i}\Ubb_{\Vcal_t}^* H  \Ubb_{\Vcal_t}\right]  \Omega \rangle\\
& = \langle \Omega, \big[ - a^*(\cc{\tilde g_1} ) a( \cc{\tilde g_2}) + \langle \cc{\tilde g_2}, \cc{\tilde g_1} \rangle, \frac{1}{i}\Ubb_{\Vcal_t}^* H  \Ubb_{\Vcal_t}\big]  \Omega \rangle\\
& = - \langle  \Omega, a^*(\cc{\tilde g_1} ) a( \cc{\tilde g_2}) \frac{1}{i}\Ubb_{\Vcal_t}^* H  \Ubb_{\Vcal_t} \Omega \rangle + \langle \Omega, \frac{1}{i}\Ubb_{\Vcal_t}^* H  \Ubb_{\Vcal_t} a^*(\cc{\tilde g_1} ) a( \cc{\tilde g_2})  \Omega \rangle = 0.
\end{split}
\]
(Here we made use of the fact that $\langle \cc{\tilde g_2}, \cc{\tilde g_1} \rangle$ as a complex number commutes with everything, and of the fact that any annihilation operator applied to the vacuum gives zero.)
Similarly, we find for the other diagonal block as well that it vanishes,
\[
\begin{split}
& \langle \psi^{\text{qf}}_t, \left[ A^*((\id-\Gamma_t) F_2) A((\id-\Gamma_t) F_1), \frac{1}{i}H \right] \psi^{\text{qf}}_t \rangle
= 0.	\qedhere
\end{split}
\]
\end{proof}
Using the Wick theorem and the CAR, it is a simple calculation that the quasifree reduction principle \eqref{eq:quasifreethm}, applied to the Hamiltonian \eqref{eq:hamiltonian}, yields the time-dependent Bogoliubov--de\,Gennes equations \eqref{eq:TDBCS}.
\begin{rem}
 The reader may wonder how it is possible that the many-body evolution gives rise to the equation
 \begin{equation}\label{eq:1}\langle F_1, \left(\partial_t \Gamma^{\text{MB}}_t\right) F_2\rangle = \langle \psi_t, [A^*(F_2) A(F_1),\frac{1}{i}H] \psi_t \rangle\end{equation}                                                                                                                                                                                                                  
 and the effective evolution solves the seemingly identical equation
  \begin{equation}\label{eq:2} \langle F_1, \left(\partial_t \Gamma_t\right) F_2\rangle = \langle \psi^{\text{qf}}_t ,[A^*(F_2) A(F_1),\frac{1}{i}H] \psi^{\text{qf}}_t\rangle,\end{equation}
  yet the two evolutions in general differ even if they both start from quasifree initial data.
  The answer is that \eqref{eq:1} is \emph{not} a well-posed initial value problem, simply because a general Fock space state has many more degrees of freedom than just the generalized one-particle reduced density matrix; the r.\,h.\,s.\ is not a function of only $\Gamma_t$. The equation \eqref{eq:1} only makes sense if the r.\,h.\,s.\ is already prescribed by the Schr\"odinger equation \eqref{eq:SE}.
  
    On the other hand, \eqref{eq:2} is a well-defined initial value problem because  quasifree states in Fock space are (up to a phase) one-to-one with their generalized one-particle reduced density matrix. So the r.\,h.\,s.\ is a function only of $\Gamma_t$ here (alternatively think of the Wick rule which also shows that the r.\,h.\,s.\ can be expressed in terms of only $\Gamma_t$).

    We provide the rigorous proof of well-posedness for a main class of physically relevant Hamiltonians and initial data in Sect.\ \ref{sec:wellposedness}.
\end{rem}

\section{Derivation of the Quasifree Reduction Principle for Bosons}\label{sec:bosons}
In this section we present the formulation of the Dirac--Frenkel principle for one-particle reduced density matrices of bosonic systems. This is slightly more complicated than for fermionic systems because the simple projection condition has to be replaced, and because we include a condensate, but can be treated by modifications of the previously developed geometric notions.

We start by reviewing some definitions for bosonic systems where they differ from the corresponding fermionic formulas. For a comprehensive introduction we refer to \cite{Solovej}. \emph{Bosonic Fock space} is defined in the same way as for fermionic systems, simply replacing antisymmetric by symmetric wave functions:
\[\fock_\text{s} := \Cbb \oplus \bigoplus_{n=1}^\infty L^2_\text{s}(\Rbb^{3n}).\]
\emph{Creation and annihilation operators} $a^*(f)$ and $a(f)$ (where $f \in L^2(\Rbb^3)$, a one-particle wave function) are defined as
\[\begin{split}
  \left( a^*(f) \psi \right)^{(n)}(x_1,\ldots x_n) & = \frac{1}{\sqrt{n}} \sum_{j=1}^n f(x_j) \psi^{(n-1)}(x_1,\ldots ,\widehat{x_{j}},\ldots,x_n),\\
  \left( a(f) \psi \right)^{(n)}(x_1,\ldots x_n) & = \sqrt{n+1}\int \di x\, \cc{f(x)} \psi^{(n+1)}(x,x_1,\ldots,x_n).
  \end{split}
\]
The bosonic creation and annihilation operators satisfy the \emph{canonical commutation relations} (CCR), i.\,e.\
\[[a(f),a(g)] = 0,\ [a^*(f),a^*(g)]=0, \text{ and } [a(f),a^*(g)] = \langle f,g\rangle\]
for all $f,g \in L^2(\Rbb^3)$. (The definition of the commutator is $[A,B] = AB-BA$.) The corresponding operator-valued distributions satisfy the formal canonical commutation relations $[a_x,a_y]=0$, $[a^*_x,a^*_y]=0$, and $[a_x,a^*_y] = \delta(x-y)$.
\emph{Quasifree states} are defined as those states for which the \emph{Wick theorem} holds, which only differs from the fermionic case by having all positive signs, e.\,g.,
\begin{align*}\langle \psi, a^\natural_1 a^\natural_2 a^\natural_3 a^\natural_4 \psi\rangle & = \langle \psi,a^\natural_1 a^\natural_2 \psi \rangle \langle \psi, a^\natural_3 a^\natural_4 \psi \rangle + \langle \psi, a^\natural_1 a^\natural_3 \psi \rangle \langle \psi, a^\natural_2 a^\natural_4 \psi \rangle \\& \quad+ \langle \psi,a^\natural_1 a^\natural_4 \psi \rangle \langle \psi, a^\natural_2 a^\natural_3\psi \rangle.\end{align*}
The r.\,h.\,s.\ can be expressed in terms of
\[\gamma(x,y) = \langle \psi, a^*_y a_x \psi\rangle \quad \text{and}\quad \alpha(x,y) = \langle \psi, a_y a_x\psi\rangle.\]
For any bosonic quasifree state, $\gamma$ and $\alpha$ are related by
\begin{equation}\label{eq:bosonicgammaalpha}\gamma^2 + \gamma = \alpha\cc{\alpha}\,, \quad \cc{\alpha}\gamma = \cc{\gamma} \cc{\alpha}\,;\end{equation}
conversely all $\gamma$ and $\alpha$ satisfying these two equations define a (up to a phase) unique quasifree state in bosonic Fock space.

The \emph{generalized one-particle reduced density matrix} is defined as
\begin{equation}\label{eq:bosonicgeneralized1pdm} \Gamma = \left(\begin{array}{cc}
             \gamma & \alpha\\ \cc{\alpha} & 1+\cc{\gamma}
            \end{array}\right).
\end{equation}
The relations \eqref{eq:bosonicgammaalpha} characterizing it as belonging to a quasifree state can be rewritten
\begin{equation}\label{eq:bosqfprop}\Gamma \Scal \Gamma = -\Gamma, \quad \text{where }\Scal = \left(\begin{array}{cc}
             \id & 0\\ 0 & -\id
            \end{array}\right).\end{equation}
The generalized one-particle reduced density matrix $\Gamma$ is a non-negative operator on $L^2(\Rbb^3)\oplus L^2(\Rbb^3)$.

As for fermionic systems, also bosonic quasifree pure states can be written in terms of a \emph{Bogoliubov transformation} \cite{Solovej}: If $\psi \in \fock_\text{s}$ is quasifree, then there exists an implementable Bogoliubov map $\Vcal: L^2(\Rbb^3) \oplus L^2(\Rbb^3) \to L^2(\Rbb^3) \oplus L^2(\Rbb^3)$ such that $\psi = \Ubb_\Vcal\Omega$, $\Ubb_\Vcal$ being the unitary implementation of $\Vcal$. Recall that $\Ubb_\Vcal^* A(F) \Ubb_\Vcal = A(\Vcal^{-1}F)$, where the generalized creation/annihilation operators $A(F)$, $A^*(F)$, $F \in L^2(\Rbb^3) \oplus L^2(\Rbb^3)$ are defined exactly the same way as for fermions.

With $\gvac = \left( \begin{array}{cc}0 & 0\\0 & \id \end{array}\right)$ the generalized one-particle reduced density matrix of the vacuum (identical to the fermionic case), we define the spaces
\begin{align}
 \Acal & = \gvac + \vec\Acal,\quad \vec\Acal = \{ \vec\Gamma \in \Sfrak_2(L^2(\Rbb^3)\oplus L^2(\Rbb^3)): \vec\Gamma = \vec\Gamma^*\},\\
 \Acal_{+} & = \{\Gamma \in \Acal: \Gamma - \Jcal \Gamma\Jcal = -\Scal\},\\
 \Mcal & = \{\Gamma\in\Acal: \Gamma - \Jcal\Gamma\Jcal = - \Scal,\ \Gamma\Scal\Gamma = -\Gamma\}.
\end{align}
These take the role of: $\Acal$ the ambient affine space defining the scalar product, $\Acal_{+}$ the affine subspace in which the many-body evolution can be found, and $\Mcal$ the approximation manifold of generalized one-particle reduced density matrices of quasifree states. 
To see that $\Mcal$ is indeed a submanifold of $\Acal_{+}$, notice that by \eqref{eq:bosonicgammaalpha} (or by $(2\Gamma+\Scal)\Scal(2\Gamma+\Scal)=\Scal$, which is equivalent to \eqref{eq:bosqfprop}), we can write every $\Gamma \in \Mcal$ as
\[\Gamma = \Gamma(\alpha) := \begin{pmatrix}
                   \tfrac{1}{2}(\sqrt{1+4\alpha\overline{\alpha}}-1) & \alpha \\ \overline{\alpha} & \tfrac{1}{2}(\sqrt{1+4\overline{\alpha}\alpha}+1)
                  \end{pmatrix},
\]
and thus $\Mcal$ as a graph.
Alas! The computation of the tangent spaces of $\Mcal$ from its graph representation involves the derivative of the
operator square root around $\id$, which leads to Lyapunov equations of type $\{X,A\}=B$ with $A=\sqrt{1+\alpha\overline{\alpha}}$ and $B$ given. There is no simple closed formula for the solution to this equation in operator form (one can only express $X$ as a function of the eigenvalues and the eigenfunctions of $A$). We overcome this problem by noticing that it is sufficient to have a parametrization of the orthogonal complement of the tangent space.

\begin{lem}[Tangent Space, Bosonic Quasifree States]\label{lem:bosonicqftangentspace}
Let $\Gamma \in \Mcal$ and $P := -\Gamma \Scal$, and $\vec\Gcal := \{\vec{P} \in \Sfrak_2(L^2(\Rbb^3)\oplus L^2(\Rbb^3)): \Scal\vec{P}^*\Scal = \vec{P} \}$. Consider the decomposition $T_\Gamma \Acal_{+} = T_\Gamma\Mcal \oplus (T_\Gamma\Mcal)^\perp$. Then
\[(T_\Gamma\Mcal)^\perp = \left\{ -\left( P^* B P^* + (1-P^*)B(1-P^*) \right)\Scal : B \in \vec\Gcal,\ B +\Jcal B\Jcal =0\right\}.\]
\end{lem}
\begin{proof}
Let us define the map $\psi: \Gamma \mapsto P = -\Gamma \Scal$. We can explicitly write down its inverse: $\Gamma = -P\Scal$ since $\Scal^2 = \id$. Let us specify domains and codomains. In parallel to the spaces $\Acal$, $\Acal_{+}$ and $\Mcal$ we introduce (notice that $P_\text{vac}:=\psi(\gvac) = \gvac$)
\begin{align*}
 \Gcal & := P_\text{vac} + \vec\Gcal,\quad \vec\Gcal := \{\vec{P} \in \Sfrak_2(L^2(\Rbb^3)\oplus L^2(\Rbb^3)): \Scal\vec{P}^*\Scal = \vec{P} \},\\
 \Gcal_{+} & := \left\{ P \in \Gcal: P + \Jcal P\Jcal = \id \right\},\\
 \Mcal_{\Gcal} & := \left\{ P\in\Gcal: P^2 = P,\ P+\Jcal P \Jcal = \id \right\}.
\end{align*}
It is easy to check that $\psi$ is an isometric isomorphism $\Acal \to \Gcal$ (both sides with the Hilbert--Schmidt scalar product) and is also an isometric isomorphism $\Acal_{+} \to \Gcal_{+}$. Furthermore, it is a diffeomorphism $\Mcal \to \Mcal_\Gcal$.

Following the strategy of Lemma \ref{lem:tangentspacenopairing} and Lemma \ref{lem:pairingtangentspace}, with the condition $\Scal P^* \Scal = P$ taking the place of self-adjointness everywhere, we obtain
\[T_P \Mcal_\Gcal = \left\{ B \in \vec\Gcal: PBP=0=(1-P)B(1-P) \text{ and }B+\Jcal B\Jcal = 0 \right\}.\]

The differential of $\psi$ is given by $D_\Gamma\psi B = -B\Scal$, which is also an isometric isomorphism $T_\Gamma \Acal_{+} \to T_P \Gcal_{+}$. In particular it conserves orthogonality, so it is also an isomorphism
\[D_\Gamma\psi: (T_\Gamma \Mcal)^\perp \to (T_P \Mcal_\Gcal)^\perp,\]
where the orthogonal complement is defined by the decomposition $T_P\Gcal_{+} = T_P\Mcal_\Gcal \oplus (T_P\Mcal_\Gcal)^\perp$.
Rewriting \[T_P \Mcal_\Gcal = \left\{PB(1-P) + (1-P)BP: B \in T_P\Gcal_{+}\right\}\] we easily find $(T_P \Mcal_\Gcal)^\perp = \left\{ P^*B P^* + (1-P^*)B(1-P^*): B \in T_P\Gcal_{+}\right\}$. Consequently we get
\begin{align*}(T_\Gamma\Mcal)^\perp & = (D_\Gamma\psi)^{-1} (T_P\Mcal_\Gcal)^\perp \\ & = \left\{ -\left( P^*B P^* + (1-P^*)B(1-P^*)\right)\Scal: B \in T_P\Gcal_{+}\right\}.\end{align*}
Noticing that $T_P\Gcal_{+} = \{B \in \vec\Gcal: B+\Jcal B\Jcal=0\}$, the proof is complete.
\end{proof}

Unlike fermionic states, bosonic states can exhibit \emph{condensation}, so that for $\psi \in \fock_s$ it is possible that for some $f \in L^2(\Rbb^3)$ we have the additional degree of freedom
\[\langle \psi, a(f) \psi\rangle \neq 0.\]
(For any quasifree state this is vanishing.) Let us define the \emph{Weyl operator}
\[\Wbb(\varphi) := \exp(a(\varphi)-a^*(\varphi)), \quad f \in L^2(\Rbb^3).\]
The Weyl operator is unitary and $\Wbb(\varphi)^* = \Wbb(-\varphi)$; furthermore they also  satisfy $\Wbb(\varphi_1) \Wbb(\varphi_2) = \Wbb(\varphi_1+\varphi_2) e^{-i \Im \langle \varphi_1,\varphi_2\rangle}$ for all $\varphi_1, \varphi_2 \in L^2(\Rbb^3)$. An ideal condensate is described by a \emph{coherent state} $\Psi = \Wbb(\varphi)\Omega$; we have 
\[\Wbb(\varphi)^* a(f) \Wbb(\varphi) = a(f) + \langle g,f\rangle, \quad \Wbb(\varphi)^* a^*(f) \Wbb(\varphi) = a^*(f) + \langle f,g\rangle\]
and consequently $\langle \Psi, a(f) \Psi\rangle = \langle f,\varphi\rangle$. The expected number of particles in the coherent state is $\langle \Psi,\Ncal \Psi\rangle = \norm{\varphi}_{L^2}^2$. Using the BCH formula\footnote{The BCH formula states that for any two operators $A$, $B$ which both commute with $[A,B]$, we have $e^{A+B} = e^{-\frac{1}{2}[A,B]} e^A e^B$.} together with the CCR we find
\begin{align*}\Psi & = e^{-\frac{1}{2}\norm{\varphi}_{L^2}^2} e^{a^*(\varphi)} e^{a(\varphi)} \Omega = e^{-\frac{1}{2}\norm{\varphi}_{L^2}^2} \sum_{n=0}^\infty \frac{a^*(\varphi)^n}{n!}\Omega \\ &= e^{-\frac{1}{2}\norm{\varphi}_{L^2}^2} \sum_{n=0}^\infty \frac{1}{\sqrt{n!}}f\otimes \cdots \otimes f.\end{align*}
From the last formula we see that a coherent state is a linear combination of different particle numbers, where the probability to measure $n$ particles is given by a Poisson distribution peaked at the value $\norm{\varphi}_{L^2}^2$.

We now enlarge the class of quasifree states to the class of Bogoliubov states by including a condensate; more precisely, a \emph{Bogoliubov state}\footnote{A remark on the nomenclature: In the literature often also states of the form $\Wbb(\varphi) \Ubb_\Vcal \Omega$ are called quasifree states. We prefer to call them Bogoliubov states, to distinguish them from quasifree states $\Ubb_\Vcal \Omega$ which satisfy the Wick rule as given before and \eqref{eq:bosqfprop}.} is any state of the form
\begin{equation}\label{eq:bogoliubovstate}\Wbb(\varphi) \Ubb_\Vcal \Omega\end{equation}
where $\varphi \in L^2(\Rbb^3)$ (typically not normalized) and $\Vcal$ is any implementable Bogoliubov map.
 Using the fact that any expectation value of an odd number of creation and annihilation operators in a quasifree state vanishes, \[\langle \Ubb_\Vcal \Omega, a^\natural_1 \cdots a^\natural_{2n+1} \Ubb_\Vcal \Omega\rangle = 0,\]
 we find
 \begin{equation}\label{eq:condensate}\langle \Wbb(\varphi) \Ubb_\Vcal \Omega, a_x \Wbb(\varphi) \Ubb_\Vcal \Omega\rangle = \langle \Ubb_\Vcal \Omega, \big(a_x + \varphi(x)\big) \Ubb_\Vcal \Omega\rangle = \varphi(x).\end{equation}
  Furthermore we find that the one-particle reduced density matrix is given by
\begin{equation}\label{eq:gammatilde}\begin{split}\gamma(x,y) &= \langle \Wbb(\varphi) \Ubb_\Vcal \Omega, a^*_y a_x \Wbb(\varphi) \Ubb_\Vcal \Omega\rangle\\
   & = \langle \Ubb_\Vcal \Omega, \big(a^*_y + \cc{\varphi(y)}\big) \big(a_x + \varphi(x)\big) \Ubb_\Vcal \Omega\rangle\\
   & = \langle \Ubb_\Vcal \Omega, a^*_y a_x \Ubb_\Vcal \Omega\rangle + \cc{\varphi(y)}\varphi(x) =: \tilde\gamma(x,y) + \cc{\varphi(y)}\varphi(x).
\end{split}\end{equation}
Similarly, we find the pairing density to be
\begin{equation}\label{eq:alphatilde}\alpha(x,y) = \langle \Ubb_\Vcal \Omega, a_y a_x \Ubb_\Vcal \Omega\rangle + \varphi(x)\varphi(y) =: \tilde\alpha(x,y) + \varphi(x)\varphi(y).\end{equation}
In other words, $\gamma = \tilde\gamma + \lvert \varphi\rangle\langle \varphi \rvert$ and $\alpha = \tilde\alpha + \varphi \otimes\varphi$. The $\tilde\gamma$ and $\tilde\alpha$ so introduced are called the \emph{truncated expectations}. They clearly satisfy the quasifree-property
\begin{equation}\label{eq:tildequasifree}\tilde\Gamma \Scal \tilde\Gamma = - \tilde\Gamma, \quad \text{where } \tilde\Gamma= \left(\begin{array}{cc}
             \tilde\gamma & \tilde\alpha\\ \cc{\tilde\alpha} & 1+\cc{\tilde\gamma}
            \end{array}\right).\end{equation}
So by first obtaining $\varphi$ through \eqref{eq:condensate} and then solving \eqref{eq:gammatilde} and \eqref{eq:alphatilde} for $\tilde\alpha$, $\tilde\gamma$, we have a natural way of assigning a unique $(\varphi,\tilde\gamma,\tilde\alpha)$ to every quasifree state; conversely every triple $(\varphi,\tilde\gamma,\tilde\alpha)$ satisfying \eqref{eq:tildequasifree} defines a (up to a phase) unique Bogoliubov state in Fock space through \eqref{eq:bogoliubovstate}.

So as we just argued, Bogoliubov states are characterized by independently the condensate wave function $\varphi \in L^2(\Rbb^3)$ and the truncated expectations, i.\,e.\ $\tilde\Gamma$. We therefore introduce the manifold
\[\Mcal^\text{bog} = L^2(\Rbb^3) \times \Mcal\quad \subset\quad L^2(\Rbb^3) \times \Acal,\]
where $\Mcal$ is the manifold of quasifree generalized one-particle reduced density matrices as determined before. Of course, the tangent space is given by
\begin{equation}\label{eq:bogtangentspace}T_{(\varphi,\tilde\Gamma)} \Mcal^\text{bog} = L^2(\Rbb^3) \oplus T_{\tilde\Gamma} \Mcal.\end{equation}

So we can now formulate the \emph{Dirac--Frenkel principle for the condensate wave function and the generalized reduced density matrix of bosonic Bogoliubov states}:
Calculate the derivative of the condensate wave function evolving by the many-body Hamiltonian $H$ in the Bogoliubov state associated with $\varphi_t$ and $\gtqf$, 
\[\langle f,\partial_t \varphi_t\rangle = \langle \psi^\text{bog}_t, [a(f),\frac{1}{i}H] \psi^\text{bog}_t \rangle,\]
  then apply the projection onto the tangent space to $\partial_t \varphi_t$. Calculate the derivative of the generalized one-particle reduced density matrix evolving by the many-body Hamiltonian $H$ in the quasifree state associated with $\gtqf$,
\[\langle F_1, \partial_t \gtqf F_2\rangle
  = \langle \Ubb_{\Vcal_t}\Omega, [A^*(F_2) A(F_1),\frac{1}{i}H] \Ubb_{\Vcal_t}\Omega\rangle,\]
  then apply the projection onto the tangent space to $\partial_t \gtqf$. The projected derivatives describe the effective evolution.

\begin{thm}[The Quasifree Reduction Principle for Bogoliubov States]
 The Dirac--Frenkel principle, applied by projecting the curve of the many-body evolution from the space $L^2(\Rbb^3) \times \Acal$ to the approximation manifold of Bogoliubov states  $\Mcal^\text{bog} = L^2(\Rbb^3) \times \Mcal$ yields the equations of the quasifree reduction principle \eqref{eq:bosonicqfreduction}. 
\end{thm}
As it was already the case for fermionic systems, the only difference between the Dirac--Frenkel principle and the principle of quasifree reduction is the projection onto the tangent space. So instead of really doing the projection onto the tangent space, we simply check that the right-hand sides of \eqref{eq:bosonicqfreduction} are orthogonal to the orthogonal complement of the tangent space.
\begin{proof}
Recall \eqref{eq:bogtangentspace}: as far as $\varphi$ is concerned, the tangent space is given by all of $L^2(\Rbb^3)$; i.\,e.\ the projection onto the tangent space is just the identity. Therefore we only have to take care of projecting the evolution of $\gamma$ and $\alpha$; more precisely we will check that the derivative of $\tilde \Gamma^\text{qf}_t$ already lives in the tangent space. For this, it is sufficient to show that $\langle A, \partial_t \gtqf\rangle_{\Sfrak_2} = 0$ for all operators $A \in \left( T_\Gamma \Mcal\right)^\perp$. Notice that a priori $\partial_t \gtqf \in T_\Gamma \Acal_{+}$, so it is sufficient to consider the orthogonal complement as a subspace of $T_\Gamma \Acal_{+}$ instead of all of $\vec\Acal$. So by Lemma \ref{lem:bosonicqftangentspace}, we can write $A = -\left( P^*BP^* + (1-P^*)B(1-P^*) \right)\Scal$ for some operator $B$.

Since $B$ is Hilbert--Schmidt, it has a singular value decomposition $B = \sum_{j} \cc{\lambda_j} \lvert \xi_j \rangle\langle \varphi_j\rvert$, $\lambda_j \in \Cbb$, $\varphi_j,\xi_j \in L^2(\Rbb^3)\oplus L^2(\Rbb^3)$. Thus we find
\[\begin{split}\langle A, \partial_t \gtqf\rangle_{\Sfrak_2} & = -\!\sum_j\!\lambda_j \!\left(\! \langle \xi_j, P(\partial_t \gtqf) \Scal P \varphi_j \rangle\!+\!\langle \xi_j, (1\!-\!P)(\partial_t \gtqf) \Scal (1\!-\!P) \varphi_j \rangle\!\right)\!.\end{split}\]
So it suffices that every such expectation value vanishes individually. Recall that we have $\langle F_1, \partial_t \gtqf F_2\rangle = \langle \psi_t^\text{qf}, [A^*(F_2) A(F_1),\frac{1}{i}H] \psi^\text{qf}_t\rangle$ for all $F_1, F_2 \in L^2(\Rbb^3)\oplus L^2(\Rbb^3)$. Recall also that we can write $\psi^\text{qf}_t = \Ubb_\Vcal\Omega$, where the Bogoliubov map $\Vcal$ satisfies
\[\Vcal^{-1} = \Scal\Vcal^*\Scal,\quad \Vcal^* \Scal \Vcal = \Scal\quad \text{and}\quad \Vcal^* \gtqf \Vcal = \gvac.\] Using these facts we find
\begin{equation}\label{eq:qfids}\Vcal^{-1} \Scal (1+\gtqf \Scal) = \left(\begin{array}{cc}\id&0\\0&0\end{array}\right)\Vcal^*\quad \text{and}\quad \Vcal^{-1} (1+\Scal\gtqf) = \left(\begin{array}{cc}\id&0\\0&0\end{array}\right) \Vcal^* \Scal.\end{equation}
We now check that the second kind of expectation value vanishes. For all $F_1, F_2 \in L^2(\Rbb^3)\oplus L^2(\Rbb^3)$ we have
\[\begin{split}
  & \langle F_1, (1-P)(\partial_t \gtqf) \Scal(1-P)F_2\rangle\\
  &= \langle \psi^\text{qf}_t, [A^*(\Scal(1-P)F_2) A((1-P^*)F_1),\frac{1}{i}H] \psi^\text{qf}_t\rangle\\
  & = \langle \Omega, [ A^*\left(\Vcal^{-1}\Scal(1-P)F_2\right) A\left(\Vcal^{-1}(1-P^*)F_1\right) ,\frac{1}{i}\Ubb^*_\Vcal H \Ubb_\Vcal] \Omega\rangle\\
    & = \langle \Omega, [ A^*\left(\left(\begin{array}{cc}\id&0\\0&0\end{array}\right)\Vcal^* F_2\right) A\left(\left(\begin{array}{cc}\id&0\\0&0\end{array}\right) \Vcal^* \Scal F_1\right) ,\frac{1}{i}\Ubb^*_\Vcal H \Ubb_\Vcal] \Omega\rangle
\end{split}\]
Writing $\Vcal^* F_2 = \left(\!\! \begin{array}{c} f_2 \\  g_2\end{array}\!\! \right)$ and $\Vcal^* \Scal F_1 = \left(\!\! \begin{array}{c} f_1 \\ g_1\end{array}\!\! \right)$, we find that this is
\[\langle \Omega, [A^*(\left(\!\! \begin{array}{c} f_2 \\  0\end{array}\!\! \right)) A(\left(\!\! \begin{array}{c} f_1 \\  0\end{array}\!\! \right)),\frac{1}{i}\Ubb^*_\Vcal H \Ubb_\Vcal] \Omega\rangle = \langle \Omega, [a^*(f_2) a(f_1),\frac{1}{i}\Ubb^*_\Vcal H \Ubb_\Vcal] \Omega\rangle =0.\]
A similar calculation shows that also $\langle F_1, P (\partial_t \gtqf) \Scal P F_2\rangle =0$.

So we have shown that $\partial_t \gtqf \in T_{\gtqf} \Mcal$, thus the projection on the tangent space is trivial, and the Dirac--Frenkel principle becomes the quasifree reduction principle.
\end{proof}

The calculation to obtain the explicit evolution equations from the bosonic principle of quasifree reduction \eqref{eq:bosonicqfreduction} was sketched in \cite{Bachetal}, based on the canonical commutation relations and the Wick theorem. The explicit equations are known as the Hartree--Fock--Bogoliubov equations, here written in terms of the truncated expectations $\tilde\alpha_t$ and $\tilde\gamma_t$ and the condensate wave function $\varphi_t$,
\begin{equation}
 \label{eq:TDHFB}
 \begin{split}
 i \partial_t \varphi_t & = h_\text{HFB}(\tilde\gamma_t) \varphi_t + \Pi_V(\tilde\alpha_t + \varphi_t\otimes\varphi_t) \cc{\varphi_t}\\
 i \partial_t \tilde\gamma_t & = [h_\text{HFB}\left(\tilde\gamma_t + \lvert \varphi_t\rangle\langle \varphi_t\rvert\right),\tilde\gamma_t] + \Pi_V\left(\tilde\alpha_t+ \varphi_t \otimes\varphi_t\right) \cc{\tilde\alpha_t}\\
 &\quad - \tilde\alpha_t \Pi_V\left(\tilde\alpha_t+\varphi_t \otimes\varphi_t\right)^*, \\
 i\partial_t \tilde\alpha_t & = h_\text{HFB}\left(\tilde\gamma_t+\lvert \varphi_t\rangle\langle \varphi_t\rvert\right) \tilde\alpha + \tilde\alpha \cc{h_\text{HFB}\left(\tilde\gamma_t +\lvert \varphi_t\rangle\langle\varphi_t\rvert\right)}\\
 & \quad + \Pi_V(\tilde\alpha_t+\varphi_t\otimes\varphi_t)(1+\cc{\tilde\gamma_t}) + \tilde\gamma_t \Pi_V(\tilde\alpha_t +\varphi_t\otimes\varphi_t),
 \end{split}
\end{equation}
where $h_\text{HFB}$ differs only by the sign of the exchange term from the fermionic $h_\text{HF}$:
\[h_\text{HFB}(\tilde\gamma_t) = h + V\ast \rho_{\tilde\gamma_t} + X_V(\tilde\gamma_t).\]
More compactly, \eqref{eq:TDHFB} can be written symplectically \cite[Eq.\ (41)]{Bachetal} 
\[\begin{split}i \partial_t \varphi_t & = h_\text{HFB}(\tilde\gamma_t) \varphi_t + \Pi_V(\tilde\alpha_t + \varphi_t\otimes\varphi_t) \cc{\varphi_t},\\
i\partial_t \tilde\Gamma_t &= \Scal G_{\Gamma_t} \tilde\Gamma_t - \tilde\Gamma_t G_{\Gamma_t} \Scal,\end{split}\]
where $\tilde\Gamma_t$ is the truncated generalized one-particle density matrix \eqref{eq:tildequasifree}, and $\Gamma_t$ the (non-truncated) generalized one-particle density matrix. The generalized Hartree--Fock--Bogoliubov operator is
\begin{equation}\label{eq:genhfbop}G_{\Gamma_t} = \left( \begin{array}{cc}
                          h_\text{HFB}(\gamma_t) & \Pi_V(\alpha_t)\\ \Pi_V(\alpha_t)^* &  \cc{h_\text{HFB}(\gamma_t)}
                         \end{array}
 \right).\end{equation}


\section{Well-Posedness of the Fermionic Bogoliubov--de\,Gennes Equations}\label{sec:wellposedness}
The well-posedness of the effective equation obtained from the Dirac--Frenkel principle is not automatic; however under reasonable assumptions on the interaction potential and the initial data it can be established by standard methods. Since to our knowledge there is no proof completely spelled out in the literature, in this section we give a detailed proof that the time-dependent fermionic Bogoliubov--de\,Gennes equations are well-posed. We consider only the case $h = -\Delta$ (in particular no external potential $V_\text{ext}$ is included), and we are interested in interaction potentials $V$ including the Coulomb potential.

In this section we also consider mixed states as initial data for the Bogoliubov-de
Gennes equation, i.\,e.\ generalized one-particle density matrices satisfying $0 \leq \Gamma_0 \leq \id$ (whereas before in the derivation we considered only pure states, $\Gamma_0^2 = \Gamma_0$).

Well-posedness for similar equations has been discussed before, e.\,g.\ in \cite{Hainzletal,LewLenz} for a relativistic system (which generally exhibits finite-time blow-up, and global well-posedness only for small initial data) and in \cite{Bachetal} for the bosonic Hartree--Fock--Bogoliubov equations. They are generalizations of earlier work on the Hartree-Fock equations without pairing ($\alpha=0$)  \cite{Chadam,ChadamGlassey,BovePratoFano,Isozaki}; see also \cite{Anapolitanos}. 
All these works are applications of the abstract formalism developed by Segal \cite{Segal}.

\subsection{Duhamel Formula and Integral Form}
The standard approach to show local well-posedness is through an application of the Banach fixed-point theorem to the integral equation obtained from the Duhamel formula 
(in the spirit of the Picard-Lindel\"off Theorem). This is the strategy we also follow here. There is a small complication as it does not seem possible\footnote{Let $V(x)$ be the Coulomb potential $\tfrac{1}{\lvert x\rvert}$ and $M$ the Fourier multiplier $\sqrt{1-\Delta}$. In energy space \eqref{eq:defenergyspace} the norm is given by $\norm{\gamma}_{\Ycal_1} = \norm{M\gamma M}_{\sone}$ and $\norm{\alpha}_{\Ycal_2} = \norm{\alpha(\cdot,\cdot)}_{H^1(\Rbb^3\times\Rbb^3)}$.
Consider the mild equation for $\gamma_t$: while we can control the term with $[\Vcal(\gamma_s),\gamma_s]$ with $\sup_{0\le s\le t}\norm{\gamma_s}_{\Ycal_1}^2$,
a homogeneity argument shows that we cannot control those with $-\Pi_V(\alpha_s)\cc{\alpha_s}$ and $\alpha_s\Pi_V(\overline{\alpha}_s)$ 
with $\sup_{0\le s\le t}\norm{\alpha_s}_{\Ycal_2}^2$. Indeed we have
\begin{align*}
\tr\,\partial_j\bigg[\!\int_0^t \! e^{i\Delta(t-s)} \Pi_V(\alpha_s)\cc{\alpha_s} e^{-i\Delta(t-s)} \di s\bigg]\partial_j&=\int_0^t \! \tr\!\big(e^{i\Delta(t-s)} \partial_j\Pi_V(\alpha_s)\cc{\alpha_s} \partial_je^{-i\Delta(t-s)} \big)\di s\\
									&=\int_0^t \! \tr\big(\partial_j\Pi_V(\alpha_s)\cc{\alpha_s}\partial_j\big)\di s,
\end{align*}
where we have used Fubini's theorem and the invariance of the trace under a unitary conjugation.
This expression has contributions scaling like $1/\text{[Length]}^3$ (the Coulomb potential, and one derivative contributing from each $\partial_j$) and is of degree 2 in $\alpha_s$. 
However, a bound in terms of $\norm{\alpha_s}_{\Ycal_2}^2$ can accommodate at most two derivatives, i.\,e.\ has contributions scaling at most like $1/\text{[Length]}^2$. The same holds for the term corresponding to $\alpha_s\Pi_V(\overline{\alpha}_s)$. This provides us with reasonable doubts about the validity of the fixed-point scheme in energy space. To be controllable the two terms $-\Pi_V(\alpha_s)\cc{\alpha_s} $ and $\alpha_s\Pi_V(\overline{\alpha}_s)$ cannot be separated, and
the smoothing effect from the conjugation by $e^{i\Delta(t-s)}$ does not seem to help to estimate trace-class norms.
} 
to apply the Banach fixed-point theorem in energy space (denoted $\Ycal$ below) when $V$ has a Coulomb singularity.  
Instead we introduce an additional Banach space $\Zcal$ tailored to the Banach fixed-point theorem. 
We find local solutions in $\Zcal$ and we show afterward that these give rise to global energy space solutions.

\medskip

\noindent Recall the Bogoliubov--de\,Gennes equations from \eqref{eq:TDBCS} (with $h = -\Delta$):
\begin{equation}
 \label{eq:initialvalueproblem}
 \begin{split}
 i \partial_t \gamma_t & = [h_\text{HF}(\gamma_t),\gamma_t] - \Pi_V(\alpha_t) \cc{\alpha_t} - \alpha_t \Pi_V(\alpha_t)^*, \\
 i\partial_t \alpha_t & = h_\text{HF}(\gamma_t) \alpha + \alpha \cc{h_\text{HF}(\gamma_t)} + \Pi_V(\alpha_t)(1-\cc{\gamma_t}) - \gamma_t \Pi_V(\alpha_t)\\
 h_\text{HF}(\gamma_t) & = -\Delta + V\ast \rho_{\gamma_t} - X_V(\gamma_t) =: -\Delta + \Vcal(\gamma_t).
 \end{split}
\end{equation}
Using Duhamel's formula, we obtain the integral form\footnote{A solution to the Duhamel integral equation is called a mild solution of the corresponding differential equation; a solution of the differential equation itself is also called strong solution for emphasis.} of the Bogoliubov--de\,Gennes equations:
\begin{equation}
  \label{eq:equation_integral_form}
 \begin{split}
    \gamma_t & = e^{i\Delta t} \gamma_0 e^{-i\Delta t} - i \int_0^t \di s\, e^{i\Delta(t-s)} \Big\{ [\Vcal(\gamma_s),\gamma_s] - \Pi_V(\alpha_s) \cc{\alpha_s} \\
    & \hspace{6.2cm} + \alpha_s \Pi_V(\cc{\alpha_s}) \Big\} e^{-i\Delta(t-s)},\\
    \alpha_t & = e^{-i \hds t} \alpha_0 - i \int_0^t \di s\, e^{- i \hds (t-s)} 
	  \Big\{ \Vcal(\gamma_s)_1 \alpha_s + \Vcal(\overline{\gamma}_s)_2 \alpha_s\\ & \hspace{4.9cm}- \left[ (\gamma_s)_1 + (\overline{\gamma}_s)_2 \right] V(x_1-x_2) \alpha_s\Big\},
 \end{split}
\end{equation}
where in the equation for $\alpha_t$, we are using the identification of the Hilbert--Schmidt operator $\alpha_t$ with a two-particle wave function in $L^2(\Rbb^3\times\Rbb^3)$, 
denoting action of an operator on the $i$-th variable ($i \in \{1,2\}$) by $(\cdot)_i$, and reading $\mathds{h} := -\Delta_1 - \Delta_2 + V(x_1-x_2)$ as a two-body Schr\"odinger operator.

We denote the nonlinearities involved in \eqref{eq:equation_integral_form}, now using the operator picture instead of the wave function picture, by 
\begin{equation}\label{eq:def_K_ij}
\begin{split}
   {K}_1(\omega) & :=[\Vcal(\gamma),\gamma] - \Pi_V(\alpha) \cc{\alpha} + \alpha \Pi_V(\cc{\alpha}),\\
   {K}_2(\omega) & := \Vcal(\gamma) \alpha + \alpha \Vcal(\cc{\gamma}) - \gamma \Pi_V(\alpha) -\Pi_V(\alpha)\cc{\gamma}.
  \end{split}
\end{equation}
where $\omega$ denotes the pair $\omega = (\gamma,\alpha)$. We denote by $\Gamma = \Gamma(\gamma,\alpha) = \Gamma(\omega)$ the corresponding generalized one-particle density matrix, see \eqref{eq:generalized1pdm}.

\subsection{Choice of Banach Spaces}
 For short we write $\sone = \sone(L^2(\Rbb^3))$ for the space of trace-class operators on $L^2(\Rbb^3)$, $\stwo = \stwo(L^2(\Rbb^3))$ for that of Hilbert--Schmidt operators, 
and $\sinf = \sinf(L^2(\Rbb^3))$ for that of linear bounded operators  (equipped with the operator norm $\lVert\cdot\rVert_{\mathcal{B}}$). For another Banach space $\mathcal{X}$, the space of linear bounded operators on $\mathcal{X}$ will be written $\mathcal{B}(\mathcal{X})$ with norm $\norm{\cdot}_{\mathcal{B}(\mathcal{X})}$.

Let us introduce the Fourier multiplier $M := (1-\Delta)^{1/2}$.
We define the Banach space $\Ycal=\Ycal_1\times\Ycal_2 \subset\Sfrak_1\times\Sfrak_2$ (with norm $\norm{\cdot}_{\Ycal}$), usually called the energy space, by
\begin{equation}\label{eq:defenergyspace}\begin{split}
\Ycal_1 & :=\{\gamma\in\Sfrak_1:\gamma^*=\gamma\ \&\ \norm{\gamma}_{\Ycal_1} := \norm{M\gamma M}_{\Sfrak_1}<\infty \},\\
\Ycal_2 & :=\{\alpha\in\Sfrak_2:\alpha^T=-\alpha\ \&\ \norm{\alpha(\cdot,\cdot)}_{H^1}<\infty \}.
\end{split}\end{equation}
Here $\alpha^T$ denotes the operator with integral kernel $\alpha^T(x,y) = \alpha(y,x)$. Here and below $\norm{\alpha(\cdot,\cdot)}_{H^1}$ refers to the norm of $\alpha$ in $H^1(\Rbb^3\times\Rbb^3)$. 

We define the Banach space $\mathcal{Z}:=\mathcal{Z}_1\times\mathcal{Z}_2$ for the purpose of constructing local solutions by the Banach fixed-point theorem by
\[\begin{split}
 \mathcal{Z}_1 & :=\{\gamma\in\sone:\gamma^*=\gamma \text{ and } \norm{\gamma}_{\Zcal_1} := \norm{M\gamma}_{\Sfrak_1}+\norm{\gamma M}_{\Sfrak_1}<\infty \},\\
 \Zcal_2 & :=\{\alpha\in\stwo:\alpha^T=-\alpha \text{ and } \norm{\alpha(\cdot,\cdot)}_{H^1}<\infty \}.
 \end{split}
\]
For any operator $A$ we have
\begin{equation}\label{eq:alphanorm}\norm{A(\cdot,\cdot)}_{H^1}^2 \leq \norm{MA}_{\stwo}^2 + \norm{AM}_{\stwo}^2 \leq 2 \norm{A(\cdot,\cdot)}_{H^1}^2.\end{equation}
By \eqref{eq:alphanorm}, since the Hilbert--Schmidt norm is smaller than the trace norm, we also have $\norm{\gamma(\cdot,\cdot)}_{H^1} \leq \norm{\gamma}_{\Zcal_1}$.
To shorten notations, we sometimes identify an operator with its integral kernel and write $\norm{\alpha}_{H^1}$,
$\hds\alpha$ and $e^{it\hds}\alpha$.

\begin{lem}[Invariance under the Linear Evolution]\label{lem:invariance_under_the_flow}
  Let $V$ satisfy \eqref{eq:assump_A}, and let the two-particle Schr\"odinger operator $\hds := -\Delta_x-\Delta_y+V(x-y)$ act on $L^2(\Rbb^3\times\Rbb^3,\di x\di y)$.
  
  Then $\Ycal$ and $\Zcal$ are invariant under the group action 
  \[
   (e^{is\Delta},e^{it\hds})\cdot(\gamma,\alpha):= \Big(e^{is\Delta}\gamma e^{-is\Delta}, e^{it\hds} \alpha\Big)
  \]
  where $e^{it\hds}$ acts on the integral kernel $\alpha(\cdot,\cdot) \in L^2(\Rbb^3\times\Rbb^3)$.
\end{lem}

\begin{proof}
 The action on $\gamma$ is given by conjugation with the unitary Fourier multipliers $e^{is\Delta}$, conserving the self-adjointness and the $\sone$-, $\Zcal_1$- and $\Ycal_1$-norms. For $\alpha\in H^2(\Rbb^3\times\Rbb^3)=\dom(\hds)$, the equality $(\hds\alpha)^T=\hds (\alpha^T)$ is straightforward,
 thus the action of $e^{it\hds}$ preserves the transpose symmetry (here $\alpha^T=-\alpha$). 
The Hilbert--Schmidt norm is left invariant because $\hds$ is self-adjoint. As $V(x-y)$ is infinitesimally form-bounded w.\,r.\,t.\ $-\Delta_x - \Delta_y$, the Sobolev space $H^2(\Rbb^3\times\Rbb^3,\di x\di y)$ (being the domain of $\hds$) is invariant under $e^{it\hds}$ , and so is $H^1(\Rbb^3\times\Rbb^3,\di x\di y)$ as its form domain.
\end{proof}

\subsection{Results on Well-Posedness}\label{sec:wpres}
We now study the existence of solutions to the time-dependent Bogoliubov--de\,Gennes equations. 
As usual we expect them to conserve the number of particles $\tr(\gamma)$ and the energy of the system
\begin{equation}\label{eq:conserved}
\mathcal{E}(\Gamma):=\tr(-\Delta \gamma)+\frac{1}{2}\iint \left[\rho_{\gamma}(x)\rho_{\gamma}(y)-|\gamma(x,y)|^2+|\alpha(x,y)|^2\right]V(x-y)\di x\di y.
\end{equation}

We need the following assumptions on the potential $V$:
\begin{equation}\label{eq:assump_A}
	 V\in L^2_{\text{loc}},\quad V(-x)=V(x)\quad \text{and} \quad V^2\le C^2_V(1-\Delta),
\end{equation}
Observe that in fact the condition $V^2\le C_V^2(1-\Delta)$ \emph{imposes} $V\in L^2_{\text{loc}}$
with for all balls $B$ of size $R>1$: $\int_{B}|V|^2\le C_3C_V^2R^3$,
where $C_3>0$ only depends on the dimension. This notation of $C_V$ will be used throughout this paper.

\begin{lem}[Local Well-Posedness in $\Zcal$]\label{lem:localwp}
Assume that $V$ satisfies \eqref{eq:assump_A}. Consider a pair of initial data $(\gamma_0,\alpha_0) \in \Zcal$. Then there exists a unique $\Zcal$-continuous solution
to the mild equations \eqref{eq:equation_integral_form}.

If $[0,T)$ is the maximal interval of existence of the solution, we have the usual blow-up alternative: either $T=+\infty$ or $\lim_{t\to T^{-}}\norm{\omega_t}_{\Zcal}=\infty$.
\end{lem}

\begin{lem}[Regularity of the Solution]\label{lem:regularity}
If the pair of initial data $(\gamma_0,\alpha_0) \in \Zcal$ satisfies $[-\Delta,\gamma_0]\in\Sfrak_1$ and $\alpha_0(\cdot,\cdot)\in H^2$,
then the mild solution of the previous lemma is a strong solution
 $(\gamma_t,\alpha_t) \in C\big([0,T),\Zcal\big) \cap C^1\big([0,T),\sone\times\stwo\big)$.
\end{lem}

	The reader might wonder why we do not simply refer to \cite[Lemma~3.1]{Segal}. Transposed here, it states that if $[-\Delta,\gamma_0]\in\Zcal_1$
	and $\hds\alpha_0\in\Zcal_2$, that is $\alpha_0\in\dom_{\Sfrak_2}(|\hds|^{3/2})$, then the mild solution is a strong solution with time derivative in $\Zcal$. 
	When $V$ is the Coulomb potential, even if a kernel $\alpha(\cdot,\cdot)$ is of Schwartz class, 
	$\hds\alpha$ is generally not in $H^1(\Rbb^3\times\Rbb^3)$, unless the diagonal $\alpha(x,x)$ identically vanishes.
	Here, thanks to the transpose symmetry $\alpha_0^T=-\alpha_0$, the result still has an important value,
	but Lemma~\ref{lem:regularity} is in some sense optimal: it states regularity when the minimal requirements are satisfied. 
%
\begin{lem}[Conservation Laws]\label{lem:conserved_spectrum_and_quantities}
  Let $V$ satisfy \eqref{eq:assump_A} and let $T>0$. Let $(\omega_t)_{0\le t<T}\in C([0,T),\Zcal)$ be a solution to \eqref{eq:equation_integral_form}.
  Then the expected particle number $\tr(\gamma_t)$ is conserved, there is a unitary propagator $U(t,s)$ such that $\Gamma_t = U(t,0)\Gamma_0 U(t,0)^*$, and the spectrum of $\Gamma_t$ is conserved.

  In particular, if $0 \leq \Gamma_0 \leq \id$, then also $0 \leq \Gamma_t \leq \id$ (states remain states), 
  and if $\Gamma_0^2 = \Gamma_0$, then also $\Gamma_t^2 = \Gamma_t$ (pure quasifree states remain pure quasifree states).
  
  If the initial data $\omega_0\in\Ycal$ define a state (that is $0\le \Gamma_0\le \mathds{1}$),
  then the energy $\mathcal{E}(\Gamma_t)$ is also conserved.  
\end{lem}
\begin{rem}In the previous Lemma we have seen that the expected number of particles $\tr \gamma_t$ is always a conserved quantity. Now let us introduce the quantity
\[\norm{\Gamma_t - \Gamma_\text{vac}}_{\Sfrak_2}^2 = \tr\left(\Gamma_t-\Gamma_\text{vac}\right)^2 = 2 \tr (\gamma_t^2 -\alpha_t \cc{\alpha}) \leq 2 \tr\gamma_t.\]
It coincides with $2 \tr\gamma_t$ if and only if $\Gamma_t$ describes a pure quasifree state, and it is also a conserved quantity because $\Gamma_t$ is unitarily equivalent to $\Gamma_0$ and because there holds
\begin{equation}\label{eq:devi}\tr\left(\Gamma_t-\Gamma_\text{vac}\right)^2 - 2 \tr \gamma_t =\tr \left( \Gamma_t^2 - \Gamma_t \right).\end{equation}
We can interpret \eqref{eq:devi} as a measure of the deviation from $\Gamma_t$ being pure quasifree.

\medskip

\noindent In the bosonic case, the same role is played by \[\tr \left(P_t-P_\text{vac}\right)^2 +2 \tr \tilde \gamma_t= \tr \left(P_t^2 - P_t\right) = 2 \tr \left(\tilde\gamma_t(1+\tilde\gamma_t)-\tilde\alpha_t \cc{\tilde\alpha_t}\right),\]
where $P_t = -\tilde\Gamma_t \Scal$, $P_\text{vac} = - \Gamma_\text{vac} \Scal$, $\tilde\gamma_t$ is the truncated one-particle density matrix, and $\tilde\Gamma_t$ is the truncated generalized one-particle density matrix. To see that this is conserved, simply notice that $\tr \left( P_t^2 -P_t\right) = \tr \left( \tilde\Gamma_t \Scal \tilde\Gamma_t + \tilde\Gamma_t\right)\Scal$, and by \cite[Lemma~3.10]{Bachetal} $\tilde\Gamma_0 = \mathcal{U}_t \tilde\Gamma_t \mathcal{U}^*_t$ for some symplectomorphism $\mathcal{U}_t$ (meaning that $\mathcal{U}^*_t \Scal \mathcal{U}_t = \Scal = \mathcal{U}_t \Scal \mathcal{U}^*_t$). As a word of caution: for bosonic systems, $\tr \tilde\gamma_t$ is not conserved by itself, only the particle number $\tr \gamma_t = \tr \tilde\gamma_t +\norm{\varphi_t}^2$ is conserved.
\end{rem}

The following lemma is used, together with conservation of the energy and the particle number, 
to globally control the $\Ycal$-norm of a solution by the $\Ycal$-norm of the initial data 
and thus to ensure that it does not blow-up. 
\begin{lem}[Controlling the $\Ycal$-Norm]\label{lem:energybounds}
Let $V^2 \leq C_V^2(1-\Delta)$. Consider $\omega \in \Ycal$ satisfying $0\le \gamma^2-\alpha\overline{\alpha}\le \gamma$ (or equivalently $0 \leq \Gamma \leq \id$ for the associated $\Gamma$). Then the following (crude) estimates hold:
\begin{itemize}
\item for any $\delta > 0$ there exists a $C_\delta > 0$, depending only on $C_V$ and $\delta$, such that \[\norm{\gamma}_{\Ycal_1} \leq \frac{\Ecal(\Gamma) + (1+4\norm{\gamma}_{\Sfrak_1})C_\delta \norm{\gamma}_{\Sfrak_1}}{1-\delta (1+4 \norm{\gamma}_{\Sfrak_1})} \quad \text{and} \quad \norm{\alpha}_{\Ycal_2} \leq \sqrt{ 2 \norm{\gamma}_{\Ycal_1} }\, ;\]
\item furthermore $\lvert \Ecal(\Gamma) \rvert \leq (1+\norm{\omega}_\Ycal)^2$.
\end{itemize}
\end{lem}

The proofs of these lemmas are postponed to the following subsections. We now state the main result of this section.

\begin{thm}[Global Well-Posedness in $\Ycal$]\label{thm:mainwp} Assume that $V$ satisfies \eqref{eq:assump_A}. 
Consider a pair of initial data $(\gamma_0,\alpha_0)\in\Ycal$ satisfying $0 \leq \Gamma_0 \leq \id$.
Then there is a global mild solution $\omega \in C(\Rbb,\Ycal)$ to the Bogoliubov--de\,Gennes equations.

If additionally, $(\gamma_0,\alpha_0)$ satisfy $[\gamma_0,-\Delta] \in \Sfrak_1$, $\alpha_0(\cdot,\cdot) \in H^2$, 
then the solution is a strong solution in $C\big(\Rbb,\Ycal\big) \cap C^1\big(\Rbb,\sone\times\stwo\big)$.
\end{thm}

\subsection{Estimates on the Non-Linearities}
We state here results needed to control the operators $\Vcal(\gamma) = V\ast \rho_{\gamma} - X_V(\gamma)$ and $\Pi_V(\alpha)$ in the nonlinearities. 
	\begin{lem}\label{lem:estimate_2}
		Let $V^2\le C_V^2(1-\Delta)$. 
		For $\gamma\in\sone$ and $\alpha\in\stwo$ we have 
		\[\begin{split}\norm{X_V(\gamma)}_{\stwo}\le C_V\norm{\gamma(\cdot,\cdot)}_{H^1}, \quad \norm{X_V(\gamma)M^{-1}}_{\stwo}\le C_V\norm{\gamma}_{\stwo},\\
		   \norm{\Pi_V(\alpha)}_{\stwo}\le C_V\norm{\alpha(\cdot,\cdot)}_{H^1}, \quad \norm{\Pi_V(\alpha)M^{-1}}_{\stwo}\le C_V\norm{\alpha}_{\stwo},
		  \end{split}\]
		and for the multiplication operator $V\ast\rho_\gamma$ we have
		\begin{align*}
		&\norm{V*\rho_{\gamma}}_{\sinf}\leq C_V \norm{M^{1/2}\gamma M^{1/2}}_{\Sfrak_1}\le C_V\norm{\gamma}_{\mathcal{Z}_1},\\&  \norm{(V*\rho_{\gamma})M^{-1}}_{\sinf}\le C_V\norm{{\gamma}}_{\sone}.
		\end{align*}
	\end{lem}
	\begin{proof}[Proof of Lemma~\ref{lem:estimate_2}]
	\noindent 1. The estimate on $\norm{\Pi_V(\alpha)}_{\Sfrak_2}$ is a simple application of the operator inequality $V^2\le C_V^2(1-\Delta)$. The estimate on $\norm{X_V(\gamma)}_{\Sfrak_2}$ is the same.
	
	\medskip
	
	\noindent 2. For $\norm{V*\rho_{\gamma}}_{L^\infty}\ge \norm{V*\rho_{\gamma}}_{\mathcal{B}}$, we use the following trick explained in \cite[Section~6]{Bove2}.
	  We decompose $\gamma=a_+-a_-$, where $a_{\pm}:=M^{-1/2}(M^{1/2}\gamma M^{1/2})_{\pm}M^{-1/2}$, and $(M^{1/2}\gamma M^{1/2})_{\pm} \geq 0$ are, respectively, the positive and negative
	  part of $M^{1/2}\gamma M^{1/2}$ in its spectral decomposition. By monotonicity of the square root, and writing down the spectral decomposition of $a_{\pm}$, we obtain
	  \begin{align*}
	   \lvert V*\rho_{\gamma}\rvert & \le |V|*(\rho_{a_{+}}+\rho_{a_-})\le C_V\tr(M^{1/2}(a_{+}+a_-)M^{1/2}) \\ & =C_V\norm{M^{1/2}\gamma M^{1/2}}_{\Sfrak_1}.
	  \end{align*}
	  For $\gamma\in \mathcal{Z}_1$ self-adjoint, by splitting $\gamma = \gamma_{+} - \gamma_{-}$ (positive and negative part), we find
	    \[
	     -M^{1/2}|\gamma|M^{1/2}\le M^{1/2}\gamma M^{1/2}\le M^{1/2}|\gamma|M^{1/2}.
	    \]
	    Now with $(\lambda_i)_i$ denoting the eigenvalues of $M^{1/2}\gamma M^{1/2}$ and $(\varphi_i)_i$ a corresponding orthonormal basis of eigenvectors, we thus get
	    \begin{align*}
	       \tr \lvert M^{1/2} \gamma M^{1/2} \rvert & = \sum_{i\geq 0} \lvert \lambda_i \rvert = \sum_{i\geq 0} \max_{\pm}\,\langle \varphi_i, (\pm M^{1/2}\gamma M^{1/2})\varphi_i\rangle\\
	       & \leq \sum_{i \geq 0} \langle \varphi_i, M^{1/2} \lvert \gamma \rvert M^{1/2} \varphi_i \rangle = \tr M^{1/2} \lvert \gamma\rvert M^{1/2}.
	    \end{align*}
	    Using cyclicity of the trace in the last estimate, and afterward
	    \[\tr \lvert\gamma\rvert M \leq \norm{\lvert\gamma\rvert M}_{\Sfrak_1} = \tr (M\lvert \gamma\rvert^2 M)^{1/2} = \tr (M \gamma^2 M)^{1/2} = \norm{\gamma M}_{\Sfrak_1},\]
	    we obtain $\norm{M^{1/2}\gamma M^{1/2}}_{\Sfrak_1} \leq \norm{\gamma}_{\Zcal_1}$.
	  
	 \medskip 
	  
	 \noindent 3.  The estimate on $\norm{(V\ast \rho_\gamma)M^{-1}}_\Bcal$ is an application of the Cauchy--Schwarz inequality and Fubini--Tonelli theorem: for $\psi\in L^2(\Rbb^3)$, let $\phi:=M^{-1}\psi$. Then
	  \begin{align*}
	  & \int\bigg|\int V(x-y)\rho_{\gamma}(y)\phi(x)\di y\bigg|^2\di x \\
	  &\le \int  \Big(\int|V(x-y)|^2|\rho_{\gamma}(y)|\di y\Big)|\phi(x)|^2\di x\int|\rho_{\gamma}(y')|\di y'\\
									 &\le \norm{\rho_{\gamma}}_{L^1}\int |\rho_{\gamma}(y)|\di y\int |V(x-y)|^2|\phi(x)|^2\di x\\
									 &\le C_V^2\norm{\rho_{\gamma}}_{L^1}^2\norm{\phi}_{H^1}^2 \le C_V^2 \norm{{\gamma}}_{\Sfrak_1}^2\norm{\psi}_{L^2}^2.
	   \end{align*}
	   
	  \medskip 
	   
	 \noindent 4. Let us now prove the estimate on $\norm{X_V(\gamma)M^{-1}}_{\Sfrak_2}$ ($\norm{\Pi_V(\alpha)M^{-1}}_{\Sfrak_2}$  is estimated in the same way). We use the idea of \cite[Lemma~E.1]{Bachetal}. 
	  By cyclicity of the trace, for $A\in\Sfrak_2$, we have $\norm{A}_{\Sfrak_2}=\norm{A^*}_{\Sfrak_2}$, and it suffices to show the boundedness of
	  \[
	   \tr\Big(X_V(\gamma)M^{-2}[X_V(\gamma)]^*\Big).
	  \]
	  As $V(x)=V(-x)$, the trace is equal to
	  \begin{align*}
	    & \tr\Big(X_V(\gamma)M^{-2}[X_V(\gamma)]^*\Big) \\ &= \int \di x \iint \di y\di z\, \gamma(x,y)V(x-y)M^{-2}(y-z)\overline{\gamma(x,z)}V(x-z),
	  \end{align*}
	  where $M^{-2}(y-z)$ denotes the Yukawa potential at point $y-z$. For almost all $x$, we consider the $L^2$-function $g_x(y):=\gamma(x,y)V(x-y)$. 
	  By operator monotonicity of the inverse, we have $M^{-2}_y\le C_V^2 V(x-y)^{-2}$. We obtain the upper bound:
	  \begin{align*}
	   & \tr\Big(X_V(\gamma)M^{-2}[X_V(\gamma)]^*\Big)=\int \di x \langle M^{-2}g_x,\,g_x\rangle_{L^2(\mathbb{R}^3,\di y)}\\					
	   &\le C_V^2\int \di x \langle V(x-\cdot)^{-2}g_x,\,g_x\rangle_{L^2(\mathbb{R}^3,\di y)} \le C_V^2\iint\di x \di y |\gamma(x,y)|^2.		
	   \qedhere
	  \end{align*}
	\end{proof}    
	
	We introduce the ``polarization'' of $K_1$ and $K_2$ as a bilinear form (it is not necessary to give it a symmetrized form):
	\begin{align}
	    & K_1(\omega_{1},\omega_2) := [\Vcal(\gamma_1),\gamma_2] - \Pi_V(\alpha_1)\cc{\alpha_2} + \alpha_2\Pi_V(\cc{\alpha_1}), \label{eq:k1}\\
	    & K_2(\omega_1,\omega_2) := \Vcal(\gamma_1)\alpha_2 + \alpha_2 \Vcal(\cc{\gamma_1}) - \gamma_1 \Pi_V(\alpha_2) - \Pi_V(\alpha_2)\cc{\gamma_1}. \label{eq:k2}
	\end{align}
	
	\begin{lem}\label{lem:estimates_of_wtK}
	  Let $V$ satisfy \eqref{eq:assump_A}. Then for the nonlinearities $K_1$ and $K_2$ seen as bilinear maps (as defined below in \eqref{eq:k1} and \eqref{eq:k2}) we have the estimates
	\[\begin{split}
	  \norm{\wtK_2(\omega_1,\omega_2)}_{\Sfrak_2} & \le C\inf_{(a,b): \{a,b\} = \{1,2\}}\norm{\omega_a}_{\Sfrak_1\times\Sfrak_2} \norm{\omega_b}_{\Zcal}\, ,\\
	   \norm{K_1(\omega_1,\omega_2)}_{\Sfrak_1} & \le C\inf_{(a,b): \{a,b\} = \{1,2\}}\norm{\omega_a}_{\Sfrak_1\times\Sfrak_2}  \norm{\omega_b}_{\Zcal},
	  \end{split}\]
	  where the constant $C$ depends only on $C_V$.
	\end{lem}
	\begin{proof}
	 The estimates follow directly from Lemma~\ref{lem:estimate_2}.
	\end{proof}	

\subsection{Local Well-Posedness in $\Zcal$ (Proof of Lemma \ref{lem:localwp})}
\label{sec:localwp}
The norm on $C(I,\Zcal)$ is given by $\norm{(\gamma,\alpha)} := \sup_{t\in I} \norm{(\gamma_t,\alpha_t)}_{\Zcal}$.
    Denoting the initial values by $\gamma_0 = \ull{\gamma}$ and $\alpha_0 = \ull{\alpha}$, we define the Picard operator $\clubsuit: C(I,\Zcal) \to C(I,\Zcal)$ 
    on an interval $I$  containing $t=0$ by setting (with $K_1$ and $K_2$ as before)
\begin{equation}\label{eq:map}
 \begin{split}
    (\clubsuit \gamma)_t & := e^{i\Delta t} \ull{\gamma} e^{-i\Delta t} - i \int_0^t  e^{i\Delta(t-s)} {K}_1(\gamma_s,\alpha_s) e^{-i\Delta(t-s)} \di s,\\
    (\clubsuit \alpha)_t & := e^{-i \hds t} \ull{\alpha} - i \int_0^t  e^{- i \hds (t-s)} {K}_2(\gamma_s,\alpha_s) \di s.
 \end{split}
\end{equation}
To establish local existence we show that the nonlinearities (when simply taking the norm inside the integral and neglecting the unitaries---for the second equation 
this is possible because $\hds$ is infinitesimally operator bounded w.\,r.\,t.\ the Laplacian) are locally Lipschitz; 
then $\clubsuit$ on a sufficiently short time interval $I$ is a contraction, and a solution is found by the Banach fixed-point theorem. 
We refer to \cite[Theorem~1]{Segal} for details. 

To prove local Lipschitz continuity of the quadratic terms $K_1$ and $K_2$, it suffices to show continuity of the corresponding bilinear maps
$K_1(\cdot,\cdot)$ and $K_2(\cdot,\cdot)$ introduced in \eqref{eq:k1}-\eqref{eq:k2}. Indeed, we can write down the difference of the quadratic terms
in terms of the polarization according to the following formula:
\[
 K_j(\omega_1)-K_j(\omega_0)=K_j(\omega_1-\omega_0,\omega_1)+K_j(\omega_0,\omega_1-\omega_0),\ j\in\{1,2\}.
\]

\begin{lem}[Continuity of the Polarized Non-Linearities]
The bilinear forms $K_1(\cdot,\cdot)$ and $K_2(\cdot,\cdot)$ \eqref{eq:k1}-\eqref{eq:k2} are continuous from $\mathcal{Z}^2$ to $\Zcal_1$ and from $\mathcal{Z}^2$ to $\Zcal_2$
respectively. Their norms are bounded by a constant depending only on $C_V$.
\end{lem}

\begin{proof}
We have to estimate $\norm{K_i(\omega_1,\omega_2)}_{\Zcal_i}$ ($i=1,2$) in terms of $\norm{\omega_1}_{\Zcal}$ and $\norm{\omega_2}_{\Zcal}$.

\medskip

	 \noindent First, let us establish a formula allowing us to pull derivatives through the interaction operators. The crucial fact is that the multiplication operator $V(x-y)$ commutes with $[\nabla,\cdot]=\nabla_x-\nabla_y$ on $\Sfrak_2$, that is: $[\nabla,X_V(\gamma)]=X_V([\nabla,\gamma])$ and the analogous for  $\Pi_V(\alpha)$. Similarly we have  $\nabla (V*\rho_{\gamma})=V*\rho_{[\nabla,\gamma]}$ (with $V\ast\rho_\gamma$ read as a function; to prove this identity we use the spectral decomposition of $\gamma$) and $[\nabla,V\ast\rho_{\gamma}] = V*\rho_{[\nabla,\gamma]}$ (with $V\ast\rho_\gamma$ read as a multiplication operator on $L^2(\Rbb^3)$). We obtain the pull-through identity
	\begin{equation}\label{eq:pullthrough}
	\begin{split}
	 M\Vcal(\gamma)M^{-1} & = \frac{1-\Delta }{M}\Vcal(\gamma)M^{-1}\\
	 & =M^{-1}\Vcal(\gamma)M^{-1}-\sum_{j=1}^3\frac{\partial_j}{M}\left[\Vcal([\partial_j,\gamma])M^{-1}+\Vcal(\gamma)\frac{\partial_j}{M}\right],
	\end{split}\end{equation}
	which enables us to transfer an $M$ from the left side of $\Vcal(\gamma)$ to its right side. The same calculation holds for $\Pi_V(\alpha)$.

	\medskip
	
\noindent We can now prove continuity of the nonlinearity $K_1$. Recall that by definition \[\norm{K_1(\omega_1,\omega_2)}_{\Zcal_1} = \norm{MK_1(\omega_1,\omega_2)}_{\sone} + \norm{K_1(\omega_1,\omega_2)M}_{\sone}.\]
From \eqref{eq:k1} we get
\[\begin{split}
   & \norm{MK_1(\omega_1,\omega_2)}_{\sone}\\ & \leq \norm{M\Vcal(\gamma_1)\gamma_2}_\sone + \norm{M\gamma_2 \Vcal(\gamma_1)}_\sone + \norm{M\Pi_V(\alpha_1)\cc{\alpha_2}}_\sone + \norm{M \alpha_2 \Pi_V(\cc{\alpha_1})}_\sone\\
   & \leq \norm{M \Vcal(\gamma_1)M^{-1} M\gamma_2}_{\sone} + \norm{M\gamma_2}_{\sone} \norm{\Vcal(\gamma_1)}_\sinf\\
   & \quad + \norm{M\Pi_V(\alpha_1)M^{-1} M\cc{\alpha_2}}_\sone + \norm{M\alpha_2}_\stwo \norm{\Pi_V(\cc{\alpha_1})}_{\stwo}\\
   & \leq \left( \norm{M\Vcal(\gamma_1)M^{-1}}_\sinf + \norm{\Vcal(\gamma_1)}_\sinf \right) \norm{\gamma_2}_{\Zcal_1} \\
   & \quad + \left( \norm{M\Pi_V(\alpha_1)M^{-1}}_{\stwo} + \norm{\Pi_V(\cc{\alpha_1})}_{\stwo} \right) \norm{\alpha_2}_{H^1} .
  \end{split}
\]
Now we employ the pull-through formula \eqref{eq:pullthrough} and afterward Lemma \ref{lem:estimate_2}, as well as the fact that $M^{-1}$ and $\partial_j M^{-1}$ are bounded operators: 
\[\begin{split}
   &\norm{M\Vcal(\gamma_1)M^{-1}}_\sinf\\
   & \leq \norm{M^{-1}\Vcal(\gamma_1)M^{-1} - \sum_{j=1}^3 \frac{\partial_j}{M} \left(  \Vcal([\partial_j,\gamma_1])M^{-1} + \Vcal(\gamma_1)\frac{\partial_j}{M}\right)}_\sinf\\
   & \leq \norm{M^{-1}}_\sinf^2 \norm{\Vcal(\gamma_1)}_\sinf + \sum_{j=1}^3 \norm{\frac{\partial_j}{M}}_\sinf \norm{\Vcal([\partial_j,\gamma_1])M^{-1}}_\sinf + \sum_{j=1}^3 \norm{\frac{\partial_j}{M}}_\sinf^2 \norm{\Vcal(\gamma_1)}_\sinf\\
   & \leq C \left( \norm{\gamma_1}_{\Zcal_1} + \norm{[\partial_j,\gamma_1]}_{\sone} \right) \leq C \norm{\gamma_1}_{\Zcal_1}.
  \end{split}
\]
Similarly, using the pull-through formula for $\Pi_V(\alpha)$, we obtain
\[\norm{M\Pi_V(\alpha_1)M^{-1}}_{\stwo} \leq C \norm{\alpha}_{H^1}.\]
Combining everything we get $\norm{MK_1(\omega_1,\omega_2)}_{\sone} \leq C \norm{\omega_1}_\Zcal \norm{\omega_2}_\Zcal$.
We can estimate the term $\norm{K_1(\omega_1,\omega_2)M}_{\sone}$ in the same way. We conclude that
\[\norm{K_1(\omega_1,\omega_2)}_{\Zcal_1} \leq C \norm{\omega_1}_\Zcal \norm{\omega_2}_\Zcal.\]

\medskip

\noindent To prove continuity of the nonlinearity $K_2$, we use \eqref{eq:alphanorm} as an upper bound for $\norm{K_2}_{H^1}^2$, and then proceed by the same method as for $K_1$. We obtain
\[\norm{K_2(\omega_1,\omega_2)}_{H^1} \leq C \norm{\omega_1}_\Zcal \norm{\omega_2}_{\Zcal}. \qedhere\]
\end{proof}

\subsection{Regularity of the Solution (Proof of Lemma~\ref{lem:regularity})}

We now assume that the initial data satisfy also the additional regularity conditions 
  $[\gamma_0,-\Delta] \in \Sfrak_1$ and $\alpha_0(\cdot,\cdot)\in H^2$ (or equivalently $\alpha_0(\cdot,\cdot)\in\dom_{\Sfrak_2}(\hds)$).
  It suffices to adapt Segal's result \cite[Lemma~3.1]{Segal} to ensure that the solution is a strong solution: instead of the $\Zcal$-norm, we apply the same argument to the $\sone\times\stwo$-norm.
  
  Identifying $\gamma_t$ and $\alpha_t$ with their integral kernel, the derivatives $\dot{\gamma}_t,\dot{\alpha}_t$ are well-defined 
  as bounded linear operators from $H^2(\Rbb^3)$ to its dual $H^{-2}(\Rbb^3)$,
  and they are equal to $-i$ times the r.\,h.\,s. of \eqref{eq:initialvalueproblem}. This establishes the equations, and it just remains to prove that $\omega_t$ is 
  Fr\'echet-differentiable in $\Sfrak_1\times\Sfrak_2$.
  
  We construct the putative derivatives by the Banach fixed-point theorem. 
  We start by formally differentiating the Bogoliubov--de\,Gennes equations to see that the derivatives should satisfy the equations
  \[
  	i\dot{\gamma}_t=[-\Delta,\dot{\gamma}_t]+\tfrac{\partial}{\partial\omega}K_1(\omega_t)\dot{\omega}_t\quad\&\quad
	i\dot{\alpha}_t=\{-\Delta,\dot{\alpha}_t\}+\Pi_V(\dot{\alpha}_t)+\tfrac{\partial}{\partial\omega}K_2(\omega_t)\dot{\omega}_t,
  \]
  where $\tfrac{\partial}{\partial\omega}K_1(\omega_t)$ and $\tfrac{\partial}{\partial\omega}K_2(\omega_t)$ denotes the Fr\'echet derivative of the nonlinearities.
  As initial data for the fixed-point problem of the derivatives we have (as given by the Bogoliubov--de\,Gennes equations)
  \[
  	i\dot{\gamma}_0:=[-\Delta,\gamma_0]+K_1(\omega_0)\quad\&\quad i\dot{\alpha}_0:=\{-\Delta,\alpha_0\}+\Pi_V(\alpha_0)+K_2(\omega_0).
  \]
  We now write the equations for $\dot\gamma_t$ and $\dot\alpha_t$ in mild form:
  \begin{equation}\label{eq:mild_form_derivative}\begin{split}
  	\dot{\gamma}_t& =e^{it\Delta }\dot{\gamma}_0e^{-it\Delta}-i\int_0^te^{i(t-s)\Delta}\tfrac{\partial}{\partial\omega}K_1(\omega_s)\dot{\omega}_se^{-i(t-s)\Delta}\di s,\\
	\dot{\alpha}_t& =e^{-it\hds}\dot{\alpha}_0-i\int_0^te^{-i(t-s)\hds}\tfrac{\partial}{\partial\omega}K_2(\omega_s)\dot{\omega}_s\di s.
  \end{split}\end{equation}
  As $K_1$ and $K_2$ are quadratic functions of $\omega$, we have for all $\delta\omega \in \Zcal$
  \[
  	\tfrac{\partial}{\partial\omega}K_j(\omega_t)\delta\omega=K_j(\omega_t,\delta\omega)+K_j(\delta\omega,\omega_t), \qquad (j=1,2)
  \]
  where the $K_j(\cdot,\cdot)$'s are the (un-symmetrized) polarizations as defined in \eqref{eq:k1} and \eqref{eq:k2}.
  By Lemma~\ref{lem:estimates_of_wtK}, their extensions $\tfrac{\partial}{\partial\omega}K_j(\omega_t):\Sfrak_1\times\Sfrak_2\to \Sfrak_j$ to $\Sfrak_1\times\Sfrak_2\supset\Zcal$ 
  are continuous with norm controlled by $C_V\norm{\omega_t}_{\Zcal}$. So we can apply the Banach fixed-point theorem to \eqref{eq:mild_form_derivative} in the Banach space $\Sfrak_1\times\Sfrak_2$. We obtain a unique local solution $v_t=(g_t,a_t)_{t\in[0,T)}$.
  
  \medskip
  
  \noindent It is easy to see through a simple Gr\"onwall argument that $\norm{v_t}_{\Sfrak_1\times\Sfrak_2}$ growths at most like 
  $\exp\big(Ct\sup_{s\in[0,t]}\norm{\omega_s}_{\Zcal}\big)$, hence the maximal interval of existence of $v_t$ is the same as that of $\omega_t$.

  \medskip
  
  \noindent Following the proof of \cite[Lemma~3.1]{Segal}, we show that $w_\eps(t):=\eps^{-1}(\omega_{t+\eps}-\omega_t)-v_t$ 
  converges to $0$ in $\Sfrak_1\times\Sfrak_2$ as $\eps\to 0$ by another Gr\"onwall argument. 
  We fix $0<T_1<T$ and work on $[0,T_1]$. It is convenient to write the mild equations in terms of $\omega_s$ and $v_t$:
  \[
  \omega_t=e^{tA}\omega_0+\int_0^te^{(t-s)A}K(\omega_s)\di s\quad\&\quad v_t=e^{tA}\dot{\omega}_0+\int_0^te^{(t-s)A}\tfrac{\partial}{\partial\omega}K(\omega_s)v_s\di s,
  \]
  where $K:=(-iK_1)\times (-iK_2):\Zcal\to\Zcal$, and $A(\gamma,\alpha):=(-i[-\Delta,\gamma],-i\hds \alpha)$. A computation (including a change of variables) yields
  \begin{align*}
  	w_{\eps}(t)&=e^{tA}\big[\eps^{-1}\big(e^{\eps A}-\mathrm{id}_{\Sfrak_1\times\Sfrak_2}\big)-A \big]\omega_0\\
		&\quad +\eps^{-1}e^{tA}\bigg[\int_0^\eps e^{(\eps-s)A}K(\omega_{s})\di s-\int_0^\eps K(\omega_0)\di s\bigg]\\
		&\quad+\int_0^te^{(t-s)A} \big[\eps^{-1}(K(\omega_{s+\eps})-K(\omega_s))-\tfrac{\partial}{\partial\omega}K(\omega_s)v_s\big]\di s.
  \end{align*}
  	By the Hille-Yosida theorem the first line tends to $0$ uniformly in $t\in[0,T)$ (the operators $e^{tA}$ are unitary in $\Sfrak_1\times\Sfrak_2$). 
 	By the change of variable $\eps u=s$ the second line is
	\[
	e^{tA} \int_0^1\big[e^{\eps(1-u)A}K(\omega_{\eps u})-K(\omega_0)\big]\di u.
	\]
	By dominated convergence it converges to $0$ uniformly in $t\in[0,T_1]$. 
	Let us now deal with the third line. As $K:\Zcal\to\Zcal$ is $\norm{\cdot}_{\Zcal}$-differentiable and that $(\omega_t)_{0\le t<T}$ is $\Zcal$-valued,
	we have
	\begin{align*}
	K(\omega_{s+\eps})-K(\omega_s) & =\int_0^1\tfrac{\partial}{\partial\omega}K\big[u\,\omega_{s+\eps}+(1-u)\omega_s\big](\omega_{s+\eps}-\omega_s)\di u \\ &=: T_{s,\eps}(\omega_{s+\eps}-\omega_s),
	\end{align*}
	where $T_{s,\eps}:\Zcal\to \Zcal$ is a convex combination of the Fr\'echet derivatives in the integrand. 
	Here $K$ is quadratic, hence $T_{s,\eps}=\tfrac{1}{2}(\tfrac{\partial}{\partial\omega}K(\omega_{s+\eps})+\tfrac{\partial}{\partial\omega}K(\omega_s))$. 
	By Lemma~\ref{lem:estimates_of_wtK},
	$\norm{T_{s,\eps}}_{\mathcal{B}(\Sfrak_1\times\Sfrak_2)}$ is uniformly bounded on $[0,T_1]$, and since the convergence $\lim_{\eps\to 0}\norm{\omega_{s+\eps}-\omega_s}_{\Zcal}=0$ holds point-wise, 
	so does $\lim_{\eps\to 0}\norm{T_{s+\eps}-\tfrac{\partial}{\partial\omega}K(\omega_s)}_{\mathcal{B}(\Sfrak_1\times\Sfrak_2)}=0$. 
 	We decompose the third line:
	\begin{multline*}
		\int_0^te^{(t-s)A} \big[\eps^{-1}(K(\omega_{s+\eps})-K(\omega_s))-\tfrac{\partial}{\partial\omega}K(\omega_s)v_s\big]\di s\\
			= \int_0^te^{(t-s)A}T_{s,\eps}w_\eps(s) \di s+\int_0^t e^{(t-s)A}(T_{s,\eps}-\tfrac{\partial}{\partial\omega}K(\omega_s))v_s\di s.
	\end{multline*}
	By dominated convergence, the second integral converges to $0$ as $\eps\to 0$, uniformly in $t\in[0,T_1]$. 
	By Lemma~\ref{lem:estimates_of_wtK}, the norm of the first integral is bounded by 
	\[
	C\sup_{s\in [0,T_1]}\norm{\omega_s}_{\Zcal}\int_0^t \norm{w_\eps(s)}_{\Sfrak_1\times\Sfrak_2}\di s.
	\]
	Putting everything together, we obtain the integral inequality
	for $t\in[0,T_1]$:
	\[
	\norm{w_\eps(t)}_{\Sfrak_1\times\Sfrak_2}\le C(T_1,\eps)+C\sup_{s\in [0,T_1]}\norm{\omega_s}_{\Zcal}\int_0^t \norm{w_\eps(s)}_{\Sfrak_1\times\Sfrak_2}\di s,
	\]
	with $\lim_{\eps\to 0}C(T_1,\eps)=0$. Hence $\lim_{\eps\to 0}w_\eps(t)=0$ uniformly in $t\in[0,T_1]$, and the solution $\omega_t$ is differentiable in 
	$\Sfrak_1\times\Sfrak_2$ on $[0,T_1]$. As $0<T_1<T$ was arbitrary, this shows that the same holds on the whole interval $[0,T)$.
    

\subsection{Conservation Laws (Proof of Lemma~\ref{lem:conserved_spectrum_and_quantities})}
	\subsubsection{Existence of Unitary Propagator and Conservation of Spectrum}
	
	In this section we prove that the solution $\Gamma_t$ at time $t$ of the Bogoliubov--de\,Gennes equation is related to $\Gamma_0$ by conjugation with a unitary propagator. 
	This implies that the spectrum of $\Gamma_t$ is time-independent.
	
	\paragraph{For regular initial data} We start with the case where $[-\Delta,\gamma_t] \in \sone$ and $\alpha(\cdot,\cdot) \in H^2(\Rbb^6)$. We split $F_{\Gamma_t}$ into unbounded time-independent and bounded time-dependent part as
	\begin{align}\label{eq:splitting_of_F}
		F_{\Gamma_t}&=\begin{pmatrix}-\Delta & 0\\0 & \Delta \end{pmatrix}+\begin{pmatrix} \Vcal(\gamma_t) & \Pi_{V}(\alpha_t)\\
														-\Pi_V(\overline{\alpha}_t)& -\Vcal(\overline{\gamma_t})\end{pmatrix}=:A+B(\Gamma_t).
	\end{align}
	According to a recent reformulation \cite{Schmid2014} of the classic Kato-Yosida result \cite[Theorem~X.70]{ReedSimon2}, there exists a continuously differentiable solution to the following linear non-autonomous initial value problem (with $B_t := B(\Gamma_t)$ prescribed by the solution $\Gamma_t$ of the nonlinear Bogoliubov--de\,Gennes equations \eqref{eq:initialvalueproblem})
	\[
	\left\{
		\begin{array}{rcl}
			i\partial_t U(t,s)&=&\big(A+B_t\big)U(t,s),\\
			U(s,s)& = & \id,
		\end{array}
	\right.
	\]
	provided that the domain $D(A+B_t)$ is independent of $t$ and the function $t \mapsto (A+B_t)\varphi$ is continuously differentiable for every $\varphi \in D(A+B_t)$. The solution $U(t,s)$ then is a propagator (i.\,e.\ $U(t,s)U(s,r) = U(t,r)$ for all $r,s,t \in \Rbb$), and in particular unitary (see the proof of \cite[Theorem~X.71]{ReedSimon2}).
	
	\medskip
	
	\noindent Let us now verify the $C^1$-condition. Since we have assumed regular initial data, the mild solution $\Gamma_t$ is continuously differentiable. Since $\varphi \in H^2(\Rbb^3)$, we can insert $\id = M^{-1}M$ and obtain $\frac{\di}{\di t} (A+B_t)\varphi = B(\partial_t \Gamma_t)M^{-1} M\varphi$. Using the estimates of Lemma \ref{lem:estimate_2}, we find that $t \mapsto B(\partial_t \Gamma_t)M^{-1}$ is continuous (w.\,r.\,t.\  the operator norm).
		
	\medskip	
		
	\noindent Consider the evolution $S_t:=U(t,0)\Gamma_0U(0,t)$.
	If we prove that it satisfies the integral form of 
	the equation
	\begin{equation}\label{eq:jellyfish}
	\left\{
		\begin{array}{rcl}
			i\partial_t S_t&=&[F_{\Gamma_t},S_t],\\
			S_{t=0}&=&\Gamma_0.
		\end{array}
	\right.
	\end{equation}
	in $\mathcal{B}(L^2(\Rbb^3)^2)$, then it follows by Gr\"onwall's uniqueness argument applied in the space $\Bcal(L^2(\Rbb^3)^2)$ that $S_t$ coincides with $\Gamma_t$.
	
	To verify that $S_t$ satisfies \eqref{eq:jellyfish}, it suffices (by density) to show that the expectation value $\langle \phi,S_t\psi\rangle$ tested with functions $\psi,\phi\in H^2(\Rbb^3)^2$ satisfies the integral equation (i.\,e. a weak formulation of the integral equation). Since $H^2(\Rbb^3)^2 = D(A+B_t)$, the expectation value $\langle \phi,S_t\psi\rangle$ is differentiable with derivative $-i\langle \phi,[F_{\Gamma_t},S_t]\psi\rangle$, and by the standard Duhamel trick we find that it satisfies the integral version of \eqref{eq:jellyfish},
	\[
	 \langle \phi,S_t\psi\rangle=\langle \phi,e^{-itA}S_0e^{itA}\psi\rangle-i\int_0^t\langle \phi,e^{-i(t-s)A}[B_s,S_s]e^{i(t-s)A}\psi\rangle\di s.
	\]
	
	\paragraph{Extension to non-regular initial data}
	It remains to extend to the case when the initial data have less regularity.
	Given arbitrary $(\gamma,\alpha) \in \Zcal$, we regularize them by setting $\gamma_n := P_{-\Delta \leq n} \gamma P_{-\Delta \leq n}$ and $\alpha_n := P_{-\Delta \leq n} \alpha P_{-\Delta \leq n}$. 
	
	Consider the $\Zcal$-solutions $\Gamma_t^{(n)}$ resp.\ $\Gamma_t$ of the Bogoliubov--de\,Gennes equation with initial data $(\gamma_n,\alpha_n)$ resp.\ $(\gamma,\alpha)$.
	As $\norm{(\gamma_n,\alpha_n)}_{\Zcal}\le \norm{(\gamma,\alpha)}_{\Zcal}$, they are all defined at least on a common interval $[0,T_1]$ (the interval used for the Banach fixed-point scheme depends only on the norm of the initial data and on $C_V$).
	By a simple Gr\"{o}nwall argument they
	converge to $\Gamma_t$ in $C([0,T_1],\Zcal)$.
	
	By the argument we gave above, for the solution $\Gamma^{(n)}_t$ we have unitary propagators $U^{(n)}(t,0)$ such that $\Gamma^{(n)}_t = U^{(n)}(t,0) \Gamma^{n} U^{(n)}(0,t)$. 

	Consider the mild equations
	\[\begin{split}U^{(n)}_\text{BFP}(t) &= e^{-iAt} - i \int_0^t e^{-iA(t-s)} B(\Gamma^{(n)}_t) U^{(n)}_\text{BFP}(s) \di s,\\U_\text{BFP}(t) & = e^{-iAt} - i \int_0^t e^{-iA(t-s)} B(\Gamma_t) U_\text{BFP}(s) \di s.\end{split}\]
	Both equations are solvable on $[0,T_1]$ by applying the Banach fixed-point theorem (hence the subscript ``BFP'') in the Banach space of bounded operators, 
	and the obtained solutions are as usual unique. However, a priori we do not know that these solutions are unitaries. 
	But the solution $U^{(n)}_\text{BFP}(t)$, by local uniqueness, agrees with the unitary propagator $U^{(n)}(t,0)$ we obtained before---and thus now $U^{(n)}_\text{BFP}(t)$ is known to be unitary. Our last step is to show that $U^{(n)}_\text{BFP}(t) \to U_\text{BFP}(t)$ ($n \to \infty$) for every fixed $t$, in operator norm; this will imply the unitarity of $U_\text{BFP}(t)$. 
	
	The convergence $U^{(n)}_\text{BFP}(t) \to U_\text{BFP}(t)$ in operator norm is of course shown by a Gr\"onwall argument: writing the difference of the mild equations we find
	\[\begin{split}U^{(n)}_\text{BFP}(t) - U_\text{BFP}(t) & = -i \int_0^t e^{-iA(t-s)} B(\Gamma^{(n)}_s - \Gamma_s) U^{(n)}_\text{BFP}(s) \di s\\
	& \quad -i \int_0^t e^{-iA(t-s)} B(\Gamma_s) \left( U^{(n)}_\text{BFP}(s) - U_\text{BFP}(s)\right)\di s.\end{split}\]
	Taking the operator norm and using Lemma \ref{lem:estimates_of_wtK}, we obtain
	\[\begin{split}\norm{U^{(n)}_\text{BFP}(t) - U_\text{BFP}(t)}_{\mathcal{B}} & \leq  T_1\! \sup_{s \in [0,T_1]}\! \norm{\omega^{(n)}_s - \omega_s}_\Zcal\\ &\quad  + \!\sup_{s \in [0,T_1]}\!\norm{\omega_s}_\Zcal \!\int_0^t \norm{ U^{(n)}_\text{BFP}(s) - U_\text{BFP}(s)}_{\mathcal{B}}\di s.
	\end{split}\]
	The first summand converges to zero as $n\to\infty$. By Gr\"onwall's method, we now have $U^{(n)}(t,0) = U^{(n)}_\text{BFP}(t) \to U_\text{BFP}(t)$ in operator norm on $[0,T_1]$. We extend this to the whole maximal interval of existence $[0,T)$ of $\Gamma_t$ by repeating the same argument for each point $t\in [0,T)$ taken as initial time.
	
	Finally, this implies that, as $n\to \infty$, $U^{(n)}(t,s) = U^{(n)}_\text{BFP}(t) U^{(n)}_\text{BFP}(s)^*$ converges in operator norm to $U_\text{BFP}(t) U_\text{BFP}(s)^*$, which constitutes the intended unitary propagator $U(t,s)$.
	
	 \subsubsection{Conservation of the Particle Number $\tr(\gamma)$}	
    The conservation is easy to establish for strong solutions by differentiating the particle number $\tr \gamma_t$. In fact, consider regular initial data $(\gamma_0,\alpha_0)$, i.\,e.\ $[-\Delta,\gamma_0]\in\sone$ and $\alpha_0(\cdot,\cdot) \in H^2$. Then by Lemma \ref{lem:regularity}, we can freely differentiate, and find
    \begin{equation}\label{eq:gammaderiv}
     i\tr(\dot{\gamma}_t)=\tr[-\Delta,\gamma_t]+\tr [\Vcal(\gamma_t),\gamma_t]+\tr\left(\Pi_V(\alpha_t)\alpha_t^*-\alpha_t\Pi_V(\alpha_t^*)\right).
     \end{equation}
The first trace vanishes since it can be written as the derivative of a function which is constant due to cyclicity of the trace, i.\,e.\ 
\[\tr[-\Delta,\gamma_t] = i \frac{\di}{\di s} \tr \left( e^{is\Delta}\gamma_t e^{-is\Delta}\right)  \Big\rvert_{s=0}.\]
	The second trace vanishes by cyclicity (note that $\Vcal(\gamma_t)$ is bounded). The third trace vanishes since we can write it out as an integral and use $V(x)=V(-x)$.

We now turn to arbitrary initial data in $\Zcal$.
Since we have existence of solutions due to a Banach fixed-point argument in $\Zcal$, the solutions are continuous in $\Zcal$-norm, w.\,r.\,t.\  initial data in $\Zcal$. The number of particles $\tr \gamma_t$ is obviously $\Zcal$-continuous, and so by approximating $\Zcal$-initial data by regular initial data, $\tr \gamma_t$ is constant again.

    
    \subsubsection{Conservation of the Energy $\Ecal(\Gamma_t)$}\label{sec:conservationofenergy}
   We emphasize that 
   $\tr(-\Delta\gamma)$ is seen as the $\Ycal$-continuous functional 
   \[
   \tr(M\gamma M)-\tr(\gamma)=\int_{p\in\Rbb^3} \widehat{\gamma}(p,p)|p|^2\di p. 
   \]
   
   \paragraph{Regularization}
   Since the kinetic part of the energy functional is not $\Zcal$-continuous, we cannot use the same strategy as for particle number conservation. Instead, we introduce a regularization for which the conservation of energy holds. As before $\Pl$ denotes $\id_{(-\Delta<\Lambda)}$, and to shorten notations, we also denote $\Pl\otimes\id_{\mathbb{C}^2}$
   (acting on $L^2(\Rbb^3)^2$) by $\Pl$. We regularize both the equation and the functions:
   for any $\Lambda>0$, we consider the solution $(\Gaml_t)_{t}$ to:
   \begin{equation}\label{eq:reg_equation}
   \left\{
   	\begin{array}{rcl}
   		i\partial_t S_t&=&\big[ \Pl F_{S_t}  \Pl,\,S_t\big],\\
		S_0&=&\Pl \Gamma_0\Pl.
   	\end{array}
   \right.
   \end{equation}
  Above, $\Pl F_{S_t}  \Pl$ denotes the \emph{bounded} operator:
  \[
  \Pl F_{S_t}  \Pl=
  	\begin{pmatrix}-\Delta\Pl+\Pl\Vcal(\gamma(S_t))\Pl & \Pl \Pi_V(\alpha(S_t))\Pl \\ 
	-\Pl \Pi_V(\overline{\alpha}(S_t))\Pl & \Delta \Pl-\Pl\Vcal(\overline{\gamma}(S_t))\Pl\end{pmatrix}.
  \]
  The Duhamel form of the equation is similar to \eqref{eq:equation_integral_form}. The infinitesimal generators of the
  free evolutions of the one-body and the pairing densities are $[-\Delta\Pl,\cdot]$ and $\{-\Delta\Pl,\cdot \}+\Pl \Pi_V(\cdot)\Pl$ respectively;
  the nonlinearities $K_1^{(\Lambda)}$, $K_2^{(\Lambda)}$ are obtained from the original ones $K_1,K_2$ by replacing $\Vcal(\gamma)$
  and $\Pi_V(\alpha)$ by  $\Pl \Vcal(\gamma)\Pl$ and $\Pl \Pi_V(\alpha)\Pl$ respectively.
  
  We can apply the Banach fixed-point theorem to the regularized equations with the $\norm{\cdot}_{\Zcal}$-norm:
  the estimates are the same, and the interval of existence $[0,T)$ only depends on the initial data $\Gamma_0$, it does not 
  depend on the cutoff $\Lambda>0$. For any $\Lambda>0$, we thus obtain a solution $(\Gaml_t)_{0\le t< T}$ 
  to \eqref{eq:reg_equation} where $\Gaml_t$ is a compact perturbation of $\Pl \Gamma_{\mathrm{vac}}\Pl$.
  We write $\oml_t,\gaml_t,\alphl_t$ for the corresponding objects of the regularized solution.

  As the operator $\Pl F_{\Gaml_t}\Pl$ is bounded by $C=C(\Lambda,\norm{\oml_t}_{\Zcal})$, 
  mild solutions to \eqref{eq:reg_equation} are also strong solutions as we can differentiate in the Duhamel formulas for $\gamma$ and $\alpha$. Indeed, by the Hille-Yosida theorem the integrand is point-wise differentiable and by dominated convergence \cite[Theorem~III.6.16]{Dunford} we can differentiate inside the integral. Furthermore, by \cite[Theorem~III.6.20]{Dunford} we can pull the bounded operators $[-\Delta P_\Lambda,\cdot]$ and $\hds_\Lambda$ outside the integral, where $\hds_{\Lambda}$ denotes the operator
	\begin{equation}\label{eq:hdsdef}
	\hds_{\Lambda}\alpha:=\{-\Pl\Delta,\alpha \}+\Pl\Pi_V(\alpha)\Pl,\quad \alpha(\cdot,\cdot)\in H^2.
	\end{equation}
  
  Since the regularized equation has the same structure as the original one, conservation of the spectrum still holds (by the argument we gave before). In particular
  we have: $0\le \Gaml_t\le \id$. We show consecutively the following four points.
  \begin{enumerate}
  	\item \label{item:i} For any $\Lambda>0$, we have $\Pl \Gaml_t\Pl=\Gaml_t$ (the regularization is consistent with the evolution).
  	\item \label{item:ii} The energy $\mathcal{E}(\Gaml_t)$ and the number of particles $\tr(\gaml_t)$ are conserved.
	\item \label{item:iii} For any $T_1\in(0,T)$, $\omega_t^{(\Lambda)}$ converges to $\omega_t$ in $C([0,T_1],\sone\times\stwo)$.
	\item \label{item:iv} Let us define the potential energy as the functional
		\begin{equation}\label{eq:potenergy}
		\mathcal{E}_{\mathrm{pot}}(\Gamma) =\frac{1}{2}\Big(\tr((V*\rho_{\gamma})\gamma)-\tr(\gamma^*X_V(\gamma))+\tr(\alpha^*\Pi_V(\alpha))\Big).
	\end{equation}
	The potential energy of $\Gaml_t$ converges to that of $\Gamma_t$ on $[0,T_1]$.
  \end{enumerate}
  Points \ref{item:i} and \ref{item:iii} follow from a Gr\"onwall argument and point \ref{item:ii} follows from straightforward differentiation. 
  The conservation laws together with Lemma \ref{lem:energybounds} for the regularized solutions ensure that $(\oml_t)_{0\le t<T}$ is uniformly $\norm{\cdot}_{\Ycal}$-bounded. 
  This result and point \ref{item:iii} imply point \ref{item:iv}: the potential part of the energy converges to that of $(\omega_t)_{0\le t<T}$.

  \medskip
  
  \noindent Let us show how we can then establish the conservation of the energy. 
  The conservation of the spectrum gives $0\le g_t\le 1$ for $g_t=\gaml_t$ and $g_t=\gamma_t$. 
  Together with point \ref{item:iv} it ensures that $(\omega_t)_{0\le t<T}$ is $\Ycal$-valued. 
  Indeed for any $\Lambda_0>0$, we have by Fatou's lemma:
  \[
 	0\le  \tr(-\Delta P_{\Lambda_0}\gamma_t)\le \liminf_{\Lambda\to \infty} \tr(-\Delta\gaml_t)<\infty.
  \]
  Taking the limit $\Lambda_0\to\infty$ yields $\tr(-\Delta \gamma_t)<\infty$ by monotone convergence,
  and we obtain for all $0\le t\le T_1<T$ the inequality
  \begin{equation}\label{eq:ineq_mass}
  	\mathcal{E}(\Gamma_{t})\le\mathcal{E}(\Gamma_{0}).
  \end{equation}
  We have to show equality in \eqref{eq:ineq_mass}, i.\,e.\ that there is no loss of mass as $\Lambda\to\infty$
  (more precisely no loss of $H^1$-mass of the eigenfunctions of $\gaml_t$). 
  Equality in \eqref{eq:ineq_mass} is ensured by the time-reversal symmetry of the equation. 
  Indeed for any $0<T_1<T$, the path $t\in [0,T_1]\mapsto \overline{\Gamma}_{T_1-t}$ satisfies the same equation as $(\Gamma_t)_{0\le t<T}$, 
  hence the same arguments as above give the reverse inequality of \eqref{eq:ineq_mass}. We emphasize that the argument uses the obvious equality
  $\mathcal{E}(\Gamma)=\mathcal{E}(\overline{\Gamma})$ and the local uniqueness (due to the Banach fixed-point argument) of the solution.
  
  \paragraph{\ref{item:i} Consistency of the regularization}
	By Gr\"onwall's method we show that 
	\[f_t:=\norm{(1-\Pl)\Gaml_t}_{\Bcal}+\norm{\Gaml_t(1-\Pl)}_{\Bcal}
	\] 
	is identically zero. Rewriting \eqref{eq:reg_equation}, $\Gaml_t$ satisfies
	\begin{equation*}
		\Gaml_t=e^{-it\Pl A }\Gamma_0e^{it\Pl A }-i\int_0^t e^{-i(t-s)\Pl A } [\Pl B_s^{(\Lambda)}\Pl,\Gaml_s]e^{i(t-s)\Pl A }\di s,
	\end{equation*}
	where $A$ was defined in \eqref{eq:splitting_of_F} and $B_s^{(\Lambda)}$ denotes
	\[
	B_s^{(\Lambda)}:=\begin{pmatrix} \Vcal(\gaml_t)& \Pi_V(\alphl_t)\\ 
		-\Pi_V(\overline{\alpha}_t^{(\Lambda)})&-\Vcal(\overline{\gamma}_t^{(\Lambda)})\end{pmatrix}.
	\]
	We thus get the inequality
	\begin{align*}
		\norm{(1-\Pl)\Gaml_t}_{\Bcal}\le 0+\int_0^t \sup_{s\in [0,t]}\norm{\Pl B_s^{(\Lambda)}\Pl}_{\Bcal}f_s\di s.
	\end{align*}
	Similarly $\norm{\Gaml_t(1-\Pl)}_{\Bcal}$ satisfies the same inequality and we obtain
	\[
	f_t\le 2\sup_{s\in [0,t]}\norm{\Pl B_s^{(\Lambda)}\Pl}_{\Bcal}\int_0^t f_s\di s.
	\]
	For any $0<T_1<T$, we have 
	\[
	\sup_{s\in [0,T_1]}\norm{\Pl B_s^{(\Lambda)}\Pl}_{\Bcal}\le C\sup_{s\in [0,T_1]}\norm{\oml_s}_{\Zcal}<\infty.
	\]
	Thus $f_t=0$ on $[0,T_1]$ for all $0<T_1<T$, that is $f_t$ identically vanishes on $[0,T)$.

\paragraph{\ref{item:ii} Conservation laws for the regularized problem}
	Point \ref{item:i} ensures us that $\oml_t$ is $\Ycal$-differentiable
	on $[0,T)$, and that
	\[
	\Pl \gaml_t\Pl=\gaml_t\quad\&\quad\Pl \alphl_t\Pl=\alphl_t,\ 0\le t<T.
	\]
	We then observe that the energy is invariant under complex conjugation on $\Ycal$, 
	and can write the energy functional as
\[\Ecal(\Gamma) = \frac{1}{2}\big(\Ecal_\text{HF}(\gamma) + \Ecal_\text{HF}(\cc{\gamma})-\tr \Pi_V(\cc{\alpha})\alpha \big),\]
where $\Ecal_\text{HF}(\gamma) = \tr (-\Delta \gamma) + \frac{1}{2}\tr \left( \Vcal(\gamma)\gamma \right)$.
	Notice furthermore that due to the assumption $V(x)=V(-x)$ we have the identity
	\[
	  \tr \Pi_V(\alpha)\beta = \tr \alpha \Pi_V(\beta).
	\]
	Taking explicitly the time derivative of $\Ecal(\Gaml_t)$ and using the last identity, it is a straightforward calculation to see that
	\[\begin{split}
 	i\frac{\di}{\di t} \Ecal(\Gaml_t) & = \frac{1}{2}
  	\Big( \tr h_\text{HF}(\gaml_t) i \dot\gamma_t^{(\Lambda)} + \tr h_\text{HF}({\cc{\gamma}^{(\Lambda)}_t}) i{\dot{\cc{\gamma}}_t^{(\Lambda)}}\\
  	& \qquad\quad - \tr \Pi_V(\cc{\alpha}_t^{(\Lambda)}) i\dot{\alpha}_t^{(\Lambda)} -\tr \Pi_V(\alphl_t)i\dot{\cc{\alpha}}_t^{(\Lambda)} \Big) =0.
	\end{split}\]
	
	Conservation of the number of particles is proven similarly (and with less calculations).

\paragraph{\ref{item:iii} Convergence of the regularization }
	Let $0<T_1<T$.
	We consider the mild form of the equation on $\alphl_t$ and $\gaml_t$. 
	In the interval $[0,T_1]$ we use Gr\"onwall's method 
	and show that $g_t:=\norm{\oml_t-\omega_t}_{\Sfrak_1\times\Sfrak_2}$
	satisfies an integral inequality of the form
	\begin{equation}\label{eq:gronwall_reg}
		g_t\le C(\Lambda,T_1)+C\int_0^t g_s\di s,
	\end{equation}
	where $\lim_{\Lambda\to\infty}C(\Lambda,T_1)=0$. Point \ref{item:iii} follows from \eqref{eq:gronwall_reg}.

	Recall the definition of $\hds_\Lambda$ from \eqref{eq:hdsdef}.	From the Duhamel formula we have
	\begin{align}\alphl_t-\alpha_t&=-i\int_0^t e^{-i(t-s)\hds_{\Lambda}}\big[K_2^{(\Lambda)}(\oml_s)-K_2^{(\Lambda)}(\omega_s)\big]\di s \label{eq:split_diff1}\\
			&\quad-i\int_0^t e^{-i(t-s)\hds_{\Lambda}}\big[K_2^{(\Lambda)}(\omega_s)-K_2(\omega_s)\big]\di s \label{eq:split_diff2}\\
			&\quad -i\int_0^t \big[e^{-i(t-s)\hds_{\Lambda}}-e^{-i(t-s)\hds}\big]K_2(\omega_s)\di s \label{eq:split_diff3}\\
			&\quad+e^{-it\hds_{\Lambda}}(\Pl\alpha_0\Pl-\alpha_0)+(e^{-it\hds_{\Lambda}}-e^{-it\hds})\alpha_0.\label{eq:split_diff4}
	\end{align}
	A similar decomposition holds for $\gaml_t-\gamma_t$.
	
	\medskip
	
	\noindent Consider now the first line \eqref{eq:split_diff1}: we take its $\Sfrak_2$-norm; then by Lemma~\ref{lem:estimates_of_wtK},
	an upper bound of \eqref{eq:split_diff1} is
	\[
		C(C_V)\sup_{t\in[0,T_1]}\sup_{\omega\in\{\oml_t,\omega_t\}}\norm{\omega}_{\Zcal}\int_0^t g_s\di s,
	\]
	which gives the integral part in the integral inequality \eqref{eq:gronwall_reg}. 
	All we have to show is that the lines \eqref{eq:split_diff2} through \eqref{eq:split_diff4} also converge to $0$
	as $\Lambda\to 0$. By a similar approach we can deal with the terms of the decomposition of $\gaml_t-\gamma_t$, and both estimates will give 
	\eqref{eq:gronwall_reg}. We only estimate $\norm{\alphl_t-\alpha_t}_\stwo$ and leave $\norm{\gaml_t-\gamma_t}_\sone$ to the reader.
	
	\bigskip
	
	\noindent We first describe two technical results, which will then be useful in dealing with the remaining lines. We emphasize that $e^{-is\hds}$ and $e^{-is\hds_\Lambda}$, $s\in\Rbb$ are unitary operators which leave the Hilbert--Schmidt norm invariant. 
	For $\gaml_t-\gamma_t$, the conjugation by the unitary operators $e^{is\Delta}$ leaves the trace-norm invariant.
	
	\begin{itemize}
	\item The first technical issue is to deal with the convergence of $e^{-is\hds_{\Lambda}}$, for $s\in\Rbb$. 
	The key observation is that $\hds_{\Lambda}$ converges to
	$\hds$ in the \emph{strong-resolvent sense}. Indeed the resolvent identity gives
	\begin{align*}
	& (\hds_{\Lambda}+i)^{-1}-(\hds+i)^{-1}\\ & =(\hds_{\Lambda}+i)^{-1}\big(\{-(1-\Pl)\Delta,\cdot\}+\Pi_{V}(\cdot)-\Pl\Pi_V(\cdot)\Pl\big)(\hds+i)^{-1}.
	\end{align*}
	For $\alpha\in\Sfrak_2$, the integral kernel of $K:=(\hds+i)^{-1}\alpha$ is in $H^2(\Rbb^3\times\Rbb^3)$. By compactness of $K$, we have
	\[
	\norm{\{-(1-\Pl)\Delta,K\}}_{\Sfrak_2}\underset{\Lambda\to\infty}{\longrightarrow}0,
	\]
	Similarly, as the operator $\Pi_V(K)$ is Hilbert--Schmidt we have
	\[
	\norm{\Pi_{V}(K)-\Pl\Pi_V(K)\Pl}_{\Sfrak_2}\underset{\Lambda\to\infty}{\longrightarrow}0.
	\]
	Then \cite[Theorem~VIII.20]{ReedSimon} ensures that for any bounded Borelian function $f:\Rbb\to \mathbb{C}$, the operator
	$f(\hds_{\Lambda})$ converges to $f(\hds)$ in the strong operator topology. In particular for all $s\in\Rbb$, $e^{-is\hds_{\Lambda}}$
	converges to $e^{-is\hds}$ in the strong operator topology.
			
	\item For the second technical issue, we introduce a second level of cutoff $\Lambda'>0$. For $0\le s\le T_1$, the operator $K_2(\omega_s)$ is compact (and its integral kernel
	is in $H^1$), hence we have point-wise in $s$:
	\[
	\norm{(1-P_{\Lambda'})K_2(\omega_s)}_{\Sfrak_2}+\norm{K_2(\omega_s)(1-P_{\Lambda'})}_{\Sfrak_2}\underset{\Lambda'\to\infty}{\longrightarrow}0.
	\]
	On $[0,T_1]$, the norm $\norm{K_2(\omega_s)}_{\Sfrak_2}$ is uniformly bounded by $\sup_{t\in[0,T_1]}\norm{K_2(\omega_t)}<\infty$. 
	Hence by dominated convergence, we obtain
	\[
		\int_0^{T_1}\left(\norm{(1-P_{\Lambda'})K_2(\omega_s)}_{\Sfrak_2}+\norm{K_2(\omega_s)(1-P_{\Lambda'})}_{\Sfrak_2}\right)\di s
			\underset{\Lambda'\to\infty}{\longrightarrow}0.
	\]
	\end{itemize}
	
	\bigskip
	
	\noindent Consider now the second line \eqref{eq:split_diff2}. Taking $\Lambda'=\Lambda$, we get that it converges to $0$ as $\Lambda \to \infty$.
	
	\medskip
	
	\noindent Consider now the fourth line \eqref{eq:split_diff4}. The first term converges to $0$ since we have 
	$\lim_{\Lambda\to 0}\norm{\Pl\alpha_0\Pl-\alpha_0}_{\Sfrak_2}=0$. For its second term, splitting $\alpha_0$ into two:
	\[
	\alpha_0=P_{\Lambda'}\alpha_0P_{\Lambda'}+\big(\alpha_0-P_{\Lambda'}\alpha_0P_{\Lambda'}\big).
	\]
	The second summand vanishes as $\Lambda' \to \infty$. 
	At fixed $\Lambda'$, we then have
	\[
	\sup_{t\in[0,T]}\norm{(e^{-it\hds_{\Lambda}}-e^{-it\hds})P_{\Lambda'}\alpha_0P_{\Lambda'}}_{\Sfrak_2}\underset{\Lambda\to\infty}{\longrightarrow}0.
	\]
	This follows from the convergence in the strong operator topology of $e^{-it\hds_{\Lambda}}$ and the fact that 
	$P_{\Lambda'}\alpha_0P_{\Lambda'}\in\dom_{\Sfrak_2}(\hds)=H^2$ which gives by functional calculus the crude estimate
	\begin{equation}\label{eq:crude_est_for_el_in_dom}
	\norm{(e^{-it_2h}-e^{-it_1h})P_{\Lambda'}\alpha_0P_{\Lambda'}}_{\Sfrak_2}\le C|t_1-t_2|(\Lambda')^2\norm{\alpha_0}_{\Sfrak_2},
	\end{equation}
	where $t_1,t_2\in\Rbb$ and where $h$ denotes $\hds$ or $\hds_{\Lambda}$. 
	Hence by an $\eps/2$-argument, the fourth line converges to $0$
	as $\Lambda \to \infty$.
	
	\medskip
	
	\noindent Consider now the third line \eqref{eq:split_diff4}. As above we write
	\[
		K_2(\omega_s)=P_{\Lambda'}K_2(\omega_s)P_{\Lambda'}+(K_2(\omega_s)-P_{\Lambda'}K_2(\omega_s)P_{\Lambda'}).
	\]
	As for the second line \eqref{eq:split_diff2}, dominated convergence gives
	\[
		\int_0^{T_1}\norm{e^{-i(t-s)\hds_{\Lambda}}(K_2(\omega_s)-P_{\Lambda'}K_2(\omega_s)P_{\Lambda'})}_{\Sfrak_2}\di s
			\underset{\Lambda'\to\infty}{\longrightarrow}0.
	\]
	As $[0,T_1]$ is compact, and the bilinear map $K_2(\cdot,\cdot):\Zcal^2\to \Sfrak_2$ is continuous, 
	the map $t\in[0,T_1]\mapsto K_2(\omega_t)$ is $\norm{\cdot}_{\Sfrak_2}$-equicontinuous. 
	At fixed $s\in[0,T_1]$ and $\Lambda'>0$, the integral kernel of $P_{\Lambda'}K_2(\omega_s)P_{\Lambda'}$
	is in $H^2$, hence a similar estimate to \eqref{eq:crude_est_for_el_in_dom} holds for this operator. 
	By an $\eps/3$-argument, we get the uniform estimate
	\[
	\sup_{s,t\in[0,T_1]}\big\lVert \big[e^{-i(t-s)\hds_{\Lambda}}-e^{-i(t-s)\hds}\big]P_{\Lambda'}K_2(\omega_s)P_{\Lambda'}\big\rVert_{\Sfrak_2}
		\underset{\Lambda\to\infty}{\longrightarrow}0.
	\]
	Therefore the third line \eqref{eq:split_diff3} tends to $0$ as $\Lambda\to\infty$ 
	(by an additional $\eps/2$-argument used to choose the auxiliary cutoff level $\Lambda'>0$ at the very beginning).
	
	\paragraph{\ref{item:iv} Convergence of the potential energy}
	Recall the potential energy \eqref{eq:potenergy}. It is straightforward that it is continuous w.\,r.\,t.\ $\norm{\cdot}_{\Zcal}$. We need a little bit more.
	For $\omega_1$ and $\omega_2$ in $\Zcal$, we have:
	  \[
 	 	\begin{split}
			&\tr(V*\rho_{\gamma_2}\gamma_2)-\tr(V*\rho_{\gamma_2}\gamma_2)
				\\ & \hspace{1.5cm} = \tr(V*\rho_{\gamma_2-\gamma_1}M^{-1}M\gamma_2)+\tr(V*\rho_{\gamma_1}(\gamma_2-\gamma_1)),\\
			& \tr(\gamma_2^*X_V(\gamma_2))-\tr(\gamma_1^*X_V(\gamma_1))\\ & \hspace{1.5cm} =\tr((\gamma_2-\gamma_1)^*X_V(\gamma_2))
			+\tr(\gamma_1^*MM^{-1}X_V(\gamma_2-\gamma_1)).
		\end{split}
 	 \]
	 As in Lemma~\ref{lem:estimates_of_wtK}, by using Lemma~\ref{lem:estimate_2} we obtain the following estimate:
	 \[
	 	\big|\mathcal{E}_{\mathrm{pot}}(\Gamma_2)-\mathcal{E}_{\mathrm{pot}}(\Gamma_1)\big|
			\le C(C_V)\sup_{\omega\in\{\omega_1,\omega_2\}}\norm{\omega}_{\Zcal}\norm{\omega_2-\omega_1}_{\Sfrak_1\times\Sfrak_2},
	 \]
	 where $\Gamma_i$ denotes the generalized density matrix corresponding to $\omega_i$.
	 Point \ref{item:ii} and Point \ref{item:iii}, namely the energy conservation of $\Gaml_t$ and the convergence of $\oml_t$ to $\omega_t$ in $\Sfrak_1\times\Sfrak_2$
	 imply the convergence of $\mathcal{E}_{\mathrm{pot}}(\Gaml_t)$ to $\mathcal{E}_{\mathrm{pot}}(\Gamma_t)$.

 \subsection{Controlling the $\Ycal$-Norm (Proof of Lemma \ref{lem:energybounds})}     
 Recall the definition of the energy,
       \[\mathcal{E}(\Gamma)=\tr(-\Delta \gamma)+\frac{1}{2}\iint \big[\rho_{\gamma}(x)\rho_{\gamma}(y)-|\gamma(x,y)|^2+|\alpha(x,y)|^2\big]V(x-y)\di x\di y.\]
       Notice also that by assumption $\gamma \geq 0$, so $\norm{M^{1/2}\gamma M^{1/2}}_{\Sfrak_1} = \tr M^{1/2}\gamma M^{1/2}$.
       
       The first term in the integral (the direct term) can be estimated using Lemma \ref{lem:estimate_2}:
       \[\big| \iint \rho_{\gamma}(x)\rho_{\gamma}(y)V(x-y)\di x\di y \rvert=\lvert \tr \gamma V\ast\rho_\gamma \big| \leq C_V \norm{\gamma}_{\Sfrak_1} \norm{M^{1/2}\gamma M^{1/2}}_{\Sfrak_1}.\]
       The second term (the exchange term) can be estimated using $\lvert V(x-y) \rvert \leq C_V M$ (from $V^2 \leq C_V^2 M^2$ by operator monotonicity of the square root):
       \[\begin{split}\lvert \iint \lvert\gamma(x,y)\rvert^2 V(x-y)\di x\di y \rvert & \leq \tr \gamma \lvert V(x-y)\rvert \gamma \leq C_V \tr \gamma M \gamma \\ & = C_V \tr M^{1/2} \gamma^{1/2} \gamma \gamma^{1/2} M^{1/2}\\
       & \leq C_V \norm{\gamma}_\Bcal \tr M^{1/2} \gamma M^{1/2} \\ &\leq C_V \norm{\gamma}_{\Sfrak_1}  \norm{M^{1/2}\gamma M^{1/2}}_{\Sfrak_1}.
       \end{split}\]
         To estimate $\norm{M^{1/2}\gamma M^{1/2}}_{\Sfrak_1}$ we employ the spectral decomposition $\gamma = \sum_{i} \lambda_i \lvert \psi_i\rangle \langle \psi_i\rvert$ (where all $\lambda_i \geq 0$) to get
       \[\begin{split}\tr M^{1/2}\gamma M^{1/2} & = \sum_{i} \lambda_i \norm{M^{1/2}\psi_i}^2 \leq \sum_i \lambda_i \left[ \delta \norm{M\psi_i}^2 + \frac{\norm{\psi_i}^2}{\delta} \right] \\ & = \delta \norm{M\gamma M}_{\Sfrak_1} + \frac{\norm{\gamma}_{\Sfrak_1}}{\delta} .\end{split}\]
       
       The third term term is of the same type as the previous one, so we easily find
       \begin{align*}\big| \iint \lvert\alpha(x,y)\rvert^2 V(x-y)\di x\di y \big|  & \leq C_V \norm{\alpha}_{\Sfrak_2}  \norm{M \alpha}_{\Sfrak_2} \\ & \le\frac{C_V}{2}\big(\delta\norm{M \alpha}_{\Sfrak_2}^2+\frac{1}{\delta}\norm{\alpha}_{\Sfrak_2}^2\big) .\end{align*}
       By assumption $\alpha \alpha^*+\gamma^2 \leq \gamma $, and thus $\norm{\alpha}_{\Sfrak_2}^2 \leq \tr \gamma = \norm{\gamma}_{\Sfrak_1} $ and
       \[
       \norm{M\alpha}_{\Sfrak_2}^2 = \norm{(M\alpha)^*}_{\Sfrak_2}^2 = \tr M\alpha\alpha^* M \leq \tr M\gamma M.
       \]
       From the estimates above (and adjusting the choice of $\delta$) we conclude that
       \[\frac{\Ecal(\Gamma) + (1+4\norm{\gamma}_{\Sfrak_1})C_\delta \norm{\gamma}_{\Sfrak_1}}{1-\delta (1+4 \norm{\gamma}_{\Sfrak_1})} \geq \norm{M\gamma M}_{\Sfrak_1}.\]
    Finally, by using the symmetry $\alpha^T = -\alpha$ and going to Fourier space, we get
       \[\norm{\alpha}^2_{H^1} = \iint \di p \di q (1+p^2+q^2) \lvert \hat\alpha(p,q)\rvert^2 \leq \norm{M\alpha}_{\Sfrak_2}^2 + \norm{\alpha M}_{\Sfrak_2}^2 = 2\norm{M\alpha}_{\Sfrak_2}^2,\]
       which is estimated as above. 
       This concludes the proof of the first bound.
       
       \medskip
       
       \noindent For the second bound, by the above estimates, choosing $\delta=1$, we obtain
       \[\lvert \Ecal(\Gamma) \rvert \leq \norm{M\gamma M}_{\Sfrak_1}(1+\norm{M\gamma M}_{\Sfrak_1}) \leq (1+\norm{\omega}_\Ycal)^2.\qedhere\]
       

\subsection{Global Well-Posedness (Proof of Theorem \ref{thm:mainwp})}
  Observe that the conservation laws (Lemma \ref{lem:conserved_spectrum_and_quantities}) together with Lemma~\ref{lem:energybounds} 
  imply that the maximal interval of existence for $\Gamma_t$ is $[0,\infty)$.
  
  By energy conservation, the solution lies in $\Ycal$. Since we have conservation of the spectrum of $\Gamma_t$ we have in particular $\gamma_t \geq 0$ for all times $t$. Thus $\norm{\gamma_t}_{\Ycal_1}=\tr((1-\Delta)\gamma_t)$.
  
 We now show that $t\mapsto \gamma_t$ is $\Ycal$-continuous by checking sequential continuity. 
 Consider a sequence of times $t_n \to t_0$ ($n \to \infty$). Knowing that $t\mapsto \gamma_t$ is $\Zcal_1$-continuous 
 we conclude by the Banach--Alaoglu theorem that $\gamma_{t_n} \rightharpoonup \gamma_{t_0}$ in weak-$\ast$ topology of $\Ycal_1$. 
 Recall the Radon--Riesz property of $\sone$ \cite{Grumm}: the only thing that could go wrong is loss of mass, 
 $\tr((1-\Delta)\gamma_{t_0}) < \lim_{n \to \infty} \tr((1-\Delta)\gamma_{t_n})$. 
 To exclude such loss of mass, we write
  \[\tr((1-\Delta)\gamma_{t_n}) = \tr\gamma_{t_n} + \Ecal(\Gamma_{t_n}) - \Ecal_\text{pot}(\Gamma_{t_n}).\]
  The conservation of the energy $\Ecal(\Gamma_t)$ and of the particle number $\tr\gamma_t$ together with the $\Zcal$-continuity of $\Ecal_\text{pot}$ give
  \[\tr((1-\Delta)\gamma_{t_0}) = \lim_{n \to \infty} \tr((1-\Delta)\gamma_{t_n}).\]
  
  This concludes existence and continuity in $\Ycal$ for positive times. 
  By time-reversal symmetry (defined as in Sect.\ \ref{sec:conservationofenergy}), we also obtain the solution for negative times. 


\section*{Acknowledgements}
The authors acknowledge support by ERC Advanced grant 321029 and by VILLUM FONDEN via the QMATH Centre of Excellence (Grant No.\ 10059). The authors would like to thank S\'ebastien Breteaux, Enno Lenzmann, Mathieu Lewin and Jochen Schmid for comments and discussions about well-posedness of the Bogoliubov--de\,Gennes equations.

\bibliography{diracfrenkel}{}
\bibliographystyle{plain}

\end{document}